\documentclass[11pt]{article}%
\pdfoutput=1
\usepackage{amsfonts}
\usepackage{amsmath}
\usepackage{amssymb}
\usepackage{graphicx}
\usepackage{fullpage}%
\usepackage{hyperref}
\providecommand{\U}[1]{\protect\rule{.1in}{.1in}}
\newtheorem{theorem}{Theorem}

\newtheorem{definition}[theorem]{Definition}

\newtheorem{lemma}[theorem]{Lemma}

\newtheorem{problem}[theorem]{Problem}
\newtheorem{proposition}[theorem]{Proposition}
\newtheorem{remark}[theorem]{Remark}

\newenvironment{proof}[1][Proof]{\noindent\textbf{#1.} }{\ \rule{0.5em}{0.5em}}

\newcommand{\ket}[1]{\left\vert #1 \right\rangle}
\newcommand{\bra}[1]{\left\langle #1 \right\vert}

\begin{document}

\title{\textbf{Two-message quantum interactive proofs and the quantum separability
problem}}
\author{Patrick Hayden, Kevin Milner, and Mark M. Wilde\\\textit{School of Computer Science,}\\\textit{McGill University,}\\\textit{Montreal, Quebec H3A 2A7, Canada}}
\maketitle

\begin{abstract}
Suppose that a polynomial-time mixed-state quantum circuit, described as a
sequence of local unitary interactions followed by a partial trace, generates
a quantum state shared between two parties. One might then wonder, does this
quantum circuit produce a state that is separable or entangled? Here, we give
evidence that it is computationally hard to decide the answer to this
question, even if one has access to the power of quantum computation. We begin
by exhibiting a two-message quantum interactive proof system that can decide
the answer to a promise version of the question. We then prove that the
promise problem is hard for the class of promise problems with
\textquotedblleft quantum statistical zero knowledge\textquotedblright%
\ (\textsf{QSZK}) proof systems by demonstrating a polynomial-time Karp
reduction from the \textsf{QSZK}-complete promise problem \textquotedblleft
quantum state distinguishability\textquotedblright\ to our quantum
separability problem. By exploiting Knill's efficient encoding of a matrix
description of a state into a description of a circuit to generate the state,
we can show that our promise problem is \textsf{NP}-hard with
respect to Cook reductions. Thus, the quantum
separability problem (as phrased above) constitutes the first nontrivial
promise problem decidable by a two-message quantum interactive proof system
while being hard for both \textsf{NP} and \textsf{QSZK}. We also consider
a variant of the problem, in which a given polynomial-time mixed-state
quantum circuit accepts a quantum state as input, and the question is to
decide if there is an input to this circuit which makes its output separable
across some bipartite cut. We prove that this problem is a complete
promise problem for the class \textsf{QIP} of problems decidable by quantum
interactive proof systems. Finally, we show that a two-message quantum interactive
proof system can also decide a multipartite generalization of the quantum separability problem.

\end{abstract}

\section{Introduction}

Quantum entanglement plays a central role in quantum information science
\cite{HHHH09}. It is believed to be one of the reasons behind the
computational power of quantum algorithms \cite{EJ98}, it can enhance the
communication capacities of channels in fascinating ways
\cite{BSST99,BSST02,CLMW10,RLMKR11}, most notably as exhibited in quantum
teleportation \cite{PhysRevLett.70.1895}\ and super-dense coding
\cite{PhysRevLett.69.2881}, and it is a resource for device-independent
quantum key distribution \cite{Ekert:1991:661,VV12}. For these reasons and
others, the characterization and systematic understanding of entanglement have
long been important goals of quantum information and complexity theory.

Every quantum state has a mathematical description as a density operator
$\rho$, which is a unit-trace, positive semidefinite operator acting on some
Hilbert space $\mathcal{H}$. (In this work, we restrict ourselves to
finite-dimensional Hilbert spaces.) If the Hilbert space $\mathcal{H}$\ has a
factorization as a tensor product $\mathcal{H}_{A}\otimes\mathcal{H}_{B}$ of
two Hilbert spaces $\mathcal{H}_{A}\ $and $\mathcal{H}_{B}$, then we write the
state $\rho$ as $\rho_{AB}$ and say that it is \textit{separable} if it admits
a decomposition of the following form:%
\begin{equation}
\rho_{AB}=\sum_{z\in\mathcal{Z}}p_{Z}\left(  z\right)  \ \sigma_{A}^{z}%
\otimes\tau_{B}^{z},\label{eq:sep-state}%
\end{equation}
for collections $\left\{  \sigma_{A}^{z}\right\}  $ and $\left\{  \tau_{B}%
^{z}\right\}  $ of quantum states and some probability distribution
$p_{Z}\left(  z\right)  $ over an alphabet$~\mathcal{Z}$ \cite{W89}. Let
$\mathcal{S}$ denote the set of separable states. Often, we think of the
systems $A$ and $B$ (corresponding to Hilbert spaces $\mathcal{H}_{A}\ $and
$\mathcal{H}_{B}$) as being spatially separated, with an experimentalist Alice
possessing system $A$, while her colleague Bob possesses system $B$. The
intuition behind the above definition of separability is that a separable
state can be prepared without any quantum interaction between systems $A$
and$~B$. That is, there is an effectively classical procedure by which they
can prepare a separable state:\ Alice selects a classical variable $z$
according to $p_{Z}\left(  z\right)  $, prepares the state $\sigma^{z}$ in her
lab, and sends Bob the variable $z$ so that he can prepare the state $\tau
^{z}$ in his lab. After this process, they both discard the variable $z$, so
that (\ref{eq:sep-state}) describes their shared quantum state. It is clear
from inspecting (\ref{eq:sep-state}) that the set of separable states is a
convex set, and that any separable state has the following form as well:%
\begin{equation}
\rho_{AB}=\sum_{x\in\mathcal{X}}p_{X}\left(  x\right)  \ \left\vert \psi
_{x}\right\rangle \left\langle \psi_{x}\right\vert _{A}\otimes\left\vert
\phi_{x}\right\rangle \left\langle \phi_{x}\right\vert _{B}%
,\label{eq:separable-state-pure-decomp}%
\end{equation}
for collections $\left\{  \left\vert \psi_{x}\right\rangle _{A}\right\}  $ and
$\left\{  \left\vert \phi_{x}\right\rangle _{B}\right\}  $ of pure states and
some probability distribution $p_{X}\left(  x\right)  $ over an
alphabet$~\mathcal{X}$ of size no larger than $\dim\left(  \mathcal{H}\right)
^{2}$, the cardinality bound following from Caratheodory's
theorem \cite{H97}. Finally, if the state $\rho_{AB}$ does not admit a
decomposition of the above form, then it is \textit{entangled}---it can only
be prepared by a quantum interaction between systems $A$ and $B$.

Given the many applications of entanglement, it is clearly important to be
able to decide if a particular bipartite state is separable or entangled. When
the state is specified as the rational entries of a density matrix acting on a
finite-dimensional Hilbert space $\mathcal{H}=\mathcal{H}_{A}\otimes
\mathcal{H}_{B}$, one can formulate several variations of the problem, all of
them being known collectively as the \textit{quantum separability problem},
and characterize their computational complexity \cite{G03,G10,BCY11a}. (See
Ref.~\cite{I07}\ for a useful, though now somewhat outdated review.) Gurvits
proved that it is \textsf{NP}-hard (with respect to Cook reductions)
to decide if a state $\rho_{AB}%
\in\mathcal{S}$ or if%
\[
\min_{\sigma_{AB}\in\mathcal{S}}\left\Vert \rho_{AB}-\sigma_{AB}\right\Vert
_{2}\geq\varepsilon,
\]
where $\left\Vert A\right\Vert _{2}\equiv\sqrt{\text{Tr}\{A^{\dag}A\}}$ is the
Hilbert-Schmidt norm and $\varepsilon$ is some positive number larger than an
inverse exponential in $\dim\left(  \mathcal{H}\right)  $. Gharibian later
improved upon this result by showing that this formulation of the quantum
separability problem is strongly \textsf{NP}-hard with respect to
Cook reductions---it is still \textsf{NP}%
-hard even if $\varepsilon$ is promised to be larger than an inverse
polynomial in $\dim\left(  \mathcal{H}\right)  $. Brand\~{a}o, Christandl, and
Yard then offered a quasi-polynomial time algorithm that decides the quantum
separability problem if it is promised that $\varepsilon$ is a positive
constant~\cite{BCY11a}, by appealing to their Pinsker-inequality-like lower
bound on the squashed entanglement \cite{BCY11}\ and to the $k$-extendibility
separability test of Doherty \textit{et al}.~\cite{DPS02,DPS04}. (These
concepts are defined later in our paper.) They also considered a variant of
the promise problem where the Hilbert-Schmidt distance is replaced by the
one-way LOCC\ distance \cite{MWW09}, which characterizes the
distinguishability of $\rho_{AB}$ and $\mathcal{S}$ if Alice and Bob are
allowed to perform local operations and to send one message of classical
communication (from either Alice to Bob or Bob to Alice).

In the circuit model of quantum computation, quantum states are generated by
unitary circuits acting on some number of qubits (with some of them being
traced out in the mixed-state circuit model
\cite{Aharonov:1998:QCM:276698.276708}), and we measure the complexity of a
quantum computation by how the circuit size (number of gates and wires) scales
with the length of the input \cite{W09}. (Note that if the circuit size is
polynomial in the input length, then the number of qubits on which the circuit
acts is likewise polynomial in the input length.) Thus, from the perspective of quantum
computational complexity theory \cite{W09}, one might consider the prior
formulations of the quantum separability problem to be somewhat restrictive.
The reason is the same as that given in~\cite{R09}: the mathematical
description of a bipartite quantum state is polynomial in the dimension of the
Hilbert space, but this Hilbert space is exponential in the number of qubits
in the state. Thus, the matrix representation is exponentially larger than it
needs to be when we are in the setting of the circuit model of quantum
computation. Also, the circuit model is natural
physically, as the evolution induced by a time-varying
two-body Hamiltonian can be efficiently described by a quantum circuit \cite{berry2007efficient}.

With this dual computational and physical motivation in mind, we take an approach to the quantum separability
problem along the above lines. We suppose that we are given a description of a
quantum circuit of polynomial size\ as a sequence of quantum gates chosen from
some finite gate set, and the circuit acts on a number of qubits, each in the
state$~\left\vert 0\right\rangle $. As part of the description, the qubits are
divided into three sets:\ a set of reference system\ qubits which are traced
out, a set of qubits which belongs to Alice, and another set which belongs to
Bob. Figure~\ref{fig:circuit-for-state}\ depicts such a circuit.
\begin{figure}
[t]
\begin{center}
\includegraphics[
natheight=1.633600in,
natwidth=1.507400in,
width=1.535in
]%
{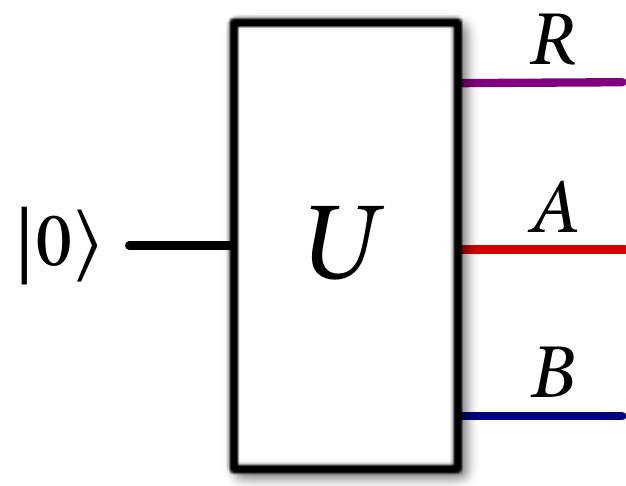}%
\caption{A unitary circuit which generates a bipartite state $\rho_{AB}$ on
systems $A$ and $B$. The qubits in the reference system $R$ are traced out.}%
\label{fig:circuit-for-state}%
\end{center}
\end{figure}
We also impose the following promise regarding the state $\rho_{AB}$ generated
by the circuit:\ there is either some state in the set of separable states
which is $\varepsilon_{1}$-close to $\rho_{AB}$ in the trace distance:%
\begin{equation}
\min_{\sigma_{AB}\in\mathcal{S}}\left\Vert \rho_{AB}-\sigma_{AB}\right\Vert
_{1}\leq\varepsilon_{1},\label{eq:YES-instance-err}%
\end{equation}
or every state in this set is at least $\varepsilon_{2}$-away from $\rho_{AB}$
in the one-way LOCC distance:%
\begin{equation}
\min_{\sigma_{AB}\in\mathcal{S}}\left\Vert \rho_{AB}-\sigma_{AB}\right\Vert
_{1-\text{LOCC}} \geq \varepsilon_{2},\label{eq:NO-instance-err}%
\end{equation}
and the difference $\varepsilon_{2}-\varepsilon_{1}$ is larger than a sufficiently large inverse
polynomial in the circuit size. The reason for the asymmetric choice of norms will become clear later in the article, but for the moment we will simply note that the problem is well-defined: the two sets identified by (\ref{eq:YES-instance-err}) and (\ref{eq:NO-instance-err}) do
not overlap because $\left\Vert \ast\right\Vert
_{1-\text{LOCC}}\leq\left\Vert \ast\right\Vert _{1}$ \cite{MWW09}. Let
\textsf{QSEP-STATE}$_{\operatorname{1,1-LOCC}}$$\left(  \varepsilon_{1},\varepsilon_{2}\right)  $ denote
this promise problem, and we will also refer to it throughout this paper as
both \textsf{QSEP-STATE}$_{\operatorname{1,1-LOCC}}$\ and the quantum separability
problem.\footnote{In an earlier version of this paper \cite{HMW13,HMW13a},
we referred to this promise problem as \textsf{QSEP-CIRCUIT}, rather
than \textsf{QSEP-STATE}$_{\operatorname{1,1-LOCC}}$. We opted for the name
change because the name \textsf{QSEP-CIRCUIT} is ambiguous: it does
not suggest that the circuit is acting to generate a quantum state (rather than acting
as an isometry or a quantum channel), nor does it indicate
that we specify the promises with respect to the trace distance and the 1-LOCC distance. In a later
work \cite{MGHW13}, we found many different entanglement detection problems that are complete
for several complexity classes in the quantum interactive proof hierarchy, and for several
of these problems, both promises are given with respect to the more natural trace distance measure.}

Quantum interactive
proof systems \cite{W03,KW00,W09} provide a  natural framework for analyzing \textsf{QSEP-STATE}$_{\operatorname{1,1-LOCC}}$. The idea is that a \textit{verifier}, who
has access to a polynomial-time bounded quantum computer, can exchange quantum
messages with a \textit{prover}, who has access to a computationally unbounded
quantum computer, in order to be convinced by the prover of the truth of some
statement. Crucial in this setting are the notions of \textit{interaction}
and\ \textit{verification}: in the case of a positive instance of a problem,
there exists some strategy by which the prover can convince the verifier to
accept with high probability, whereas in the case of a negative instance, with
high probability, there is nothing that the prover can do to convince the
verifier to accept. One appeal of the quantum setting is that it is in some ways
conceptually simpler than its classical counterpart, thanks to Kitaev and Watrous' demonstration that it is possible to parallelize
any quantum interactive proof system to at most three messages between
the verifier and prover \cite{KW00}.
Their result leaves just four natural complexity classes
derived from quantum interactive proof systems with a single prover:

\begin{itemize}
\item \textsf{QIP(3)=QIP}---the full power of quantum interactive proof systems
(known to be equivalent to \textsf{PSPACE} \cite{jain2010qip}),

\item \textsf{QIP(2)}---there are two messages exchanged (verifier-to-prover
followed by prover-to-verifier),

\item \textsf{QIP(1)}---the prover sends a quantum state to the verifier (this
class is more commonly known as \textsf{QMA}),

\item \textsf{QIP(0)}---the verifier proceeds on his own (this is equivalent
to \textsf{BQP}, the class of problems efficiently decidable on a quantum
computer with bounded error).
\end{itemize}

It has been remarked that \textsf{QIP(2)} appears to be the \textquotedblleft
most mysterious\textquotedblright\ of these complexity classes \cite{JUW09}. At the very least, not much is known about it.

\section{Overview of Results}

In this paper, we contribute the following results:

\begin{itemize}
\item \textsf{QSEP-STATE}$_{\operatorname{1,1-LOCC}}$ is decidable by a two-message quantum interactive
proof system for a wide range of parameters (this result is stated
as Theorem~\ref{thm:QSEP-in-QIP(2)}
and informally stated in the
discussion surrounding (\ref{eq:YES-instance-err})-(\ref{eq:NO-instance-err})).
The proof system builds upon the
approach of Brand\~{a}o \textit{et al}.~\cite{BCY11a} and the notion of
$k$-extendibility \cite{DPS02,DPS04}. In particular, the verifier generates
the state $\rho_{AB}$ using its description as a quantum circuit and then sends the
reference system to the prover. In the case of a positive instance, the state
is separable (or close to it), and the prover can generate a purification of a
$k$-extension of the state by acting with a unitary operation on the reference
system and some ancilla qubits. The prover sends all of his output qubits back
to the verifier, who then performs phase estimation over the symmetric group
\cite{K95} (also known as the \textquotedblleft permutation
test\textquotedblright\ \cite{BBDEJM97,KNY08})\ in order to verify whether the
state sent by the prover is a $k$-extension. This proof system has
completeness and soundness error directly related to $\varepsilon_{1}$ in
(\ref{eq:YES-instance-err}) and $\varepsilon_{2}$\ in
(\ref{eq:NO-instance-err}).

\item \textsf{QSEP-STATE}$_{\operatorname{1,1-LOCC}}$ is hard for the class of promise problems
decidable by a quantum statistical zero-knowledge (\textsf{QSZK}) proof system
\cite{W02,W09zkqa}. A \textsf{QSZK} proof system is similar to a quantum
interactive proof system, with the exception that in the case of a positive
instance of the problem, the verifier should be able to generate the state on
his systems with a polynomial-time quantum computer (so that he has not
generated any state which he could have not have generated already on his
own).\footnote{This is the definition of \textsf{QSZK}\ in the case when the verifier
behaves honestly. Watrous later showed that the expressive power of this class
is the same as the case in which the verifier does not behave honestly
\cite{W09zkqa}.} We prove this result by a reduction from the \textsf{QSZK}%
-complete promise problem \textsf{QUANTUM-STATE-DISTINGUISHABILITY}\ to
\textsf{QSEP-STATE}$_{\operatorname{1,1-LOCC}}$. This reduction is somewhat similar to a previous
reduction of Rosgen and Watrous \cite{RW05}, but the setting here is different
and thus requires a different analysis (though interestingly, the analysis is
reminiscent of arguments used in quantum information theory
\cite{ADHW06FQSW,D05}).

\item \textsf{QSEP-STATE}$_{\operatorname{1,1-LOCC}}$ is \textsf{NP}-hard with respect
to Cook reductions. This result follows from the
fact that the matrix version of the quantum separability problem is
\textsf{NP}-hard, though we require some results of Knill which show that one
can encode a quantum state efficiently by a unitary circuit if given a matrix
description of the state \cite{Kn95}.

\item We consider a variation of \textsf{QSEP-STATE}$_{\operatorname{1,1-LOCC}}$, called
\textsf{QSEP-CHANNEL}$_{\operatorname{1,1-LOCC}}$, in which some of the inputs to the circuit are
arbitrary and some are fixed in the state $\left\vert 0\right\rangle $. The
output qubits are again divided into three sets: the reference system qubits
which are traced out, Alice's qubits, and Bob's qubits. We can think of this
circuit as a quantum channel $\mathcal{N}_{S\rightarrow AB}$ with an input
system $S$ and two output systems $A$ and $B$. The task is then to decide if
there exists an input to the circuit such that the output state shared between
Alice and Bob is separable, and we have the promise that either%
\[
\min_{\rho_{S},\ \sigma_{AB}\in\mathcal{S}}\left\Vert \mathcal{N}%
_{S\rightarrow AB}\left(  \rho_{S}\right)  -\sigma_{AB}\right\Vert _{1}%
\leq\varepsilon_{1},
\]
or%
\[
\min_{\rho_{S},\ \sigma_{AB}\in\mathcal{S}}\left\Vert \mathcal{N}%
_{S\rightarrow AB}\left(  \rho_{S}\right)  -\sigma_{AB}\right\Vert
_{1-\text{LOCC}} \geq \varepsilon_{2},
\]
where again the difference $\varepsilon_{2}-\varepsilon_{1}$ is larger than a sufficiently large inverse polynomial in the circuit size. We show that this promise problem is
\textsf{QIP}-complete. The reasoning here is similar to the reasoning for our
earlier results for \textsf{QSEP-STATE}$_{\operatorname{1,1-LOCC}}$, and we again exploit the results of
Rosgen and Watrous \cite{RW05} (in particular, the fact that
\textsf{QUANTUM-CIRCUIT-DISTINGUISHABILITY}\ is \textsf{QIP}-complete). We
will also refer to this promise problem throughout as the channel quantum
separability problem.

\item \textsf{MULTI-QSEP-STATE}$_{\operatorname{1,1-LOCC}}$, a multipartite generalization of
\textsf{QSEP-STATE}$_{\operatorname{1,1-LOCC}}$, is also decidable by a two-message quantum interactive
proof system for a wide range of parameters. The analysis of the proof system
exploits a recent quantum de Finetti theorem of
Brand\~{a}o and Harrow~\cite{BH12} along with the analysis used
in the proof system for \textsf{QSEP-STATE}$_{\operatorname{1,1-LOCC}}$. \textsf{MULTI-QSEP-STATE}$_{\operatorname{1,1-LOCC}}$ is also
\textsf{NP}- and \textsf{QSZK}-hard
because \textsf{QSEP-STATE}$_{\operatorname{1,1-LOCC}}$ trivially reduces to it, as
\textsf{QSEP-STATE}$_{\operatorname{1,1-LOCC}}$ is merely a special case of \textsf{MULTI-QSEP-STATE}$_{\operatorname{1,1-LOCC}}$.
\end{itemize}

The paper is structured as follows. In the next section, we introduce
preliminary concepts such as quantum states and channels, distance measures
between quantum states, $k$-extendibility, three- and two-message quantum
interactive proof systems and complete promise problems for their
corresponding complexity classes. Section~\ref{sec:qsep-in-qip2}\ contains our
first result, in which we demonstrate that a two-message quantum interactive
proof system decides \textsf{QSEP-STATE}$_{\operatorname{1,1-LOCC}}$. We show in
Section~\ref{sec:qsep-qszk-hard}\ that \textsf{QSEP-STATE}$_{\operatorname{1,1-LOCC}}$\ is hard for the
class \textsf{QSZK}, and in Section~\ref{sec:qsep-np-hard}, we show that it is
\textsf{NP}-hard as well. In Section$~$\ref{sec:channel-qsep-qip-complete}, we
extend the above results to show that the channel variation of the quantum
separability problem is \textsf{QIP}-complete. Section~\ref{sec:qmultisep-in-qip2}
gives a two-message quantum interactive proof system that decides 
\textsf{MULTI-QSEP-STATE}$_{\operatorname{1,1-LOCC}}$. Finally, we conclude in
Section~\ref{sec:conclusion}\ with a summary and some open directions.

\section{Preliminaries}

\label{sec:prelim}In this preliminary section, we review background concepts
from quantum information theory \cite{NC00}\ and quantum computational
complexity theory \cite{W09}.

\subsection{Quantum states and channels}

A quantum state is a positive semidefinite, unit-trace operator (referred to
as the density operator)\ acting on some Hilbert space $\mathcal{H}$. Let
$\mathcal{D}(\mathcal{H})$\ denote the set of density operators acting on a
Hilbert space$~\mathcal{H}$. An \textit{extension} of a quantum state $\rho
\in\mathcal{D}(\mathcal{H}_{A})$ is some state $\omega\in\mathcal{D}%
(\mathcal{H}_{A}\otimes\mathcal{H}_{B})$ (on a larger Hilbert space) such that
$\rho=\ $Tr$_{\mathcal{H}_{B}}\left\{  \omega\right\}  $. A quantum state is
pure if its density operator is unit rank, in which case it has an
equivalent representation as a unit vector $\left\vert \psi\right\rangle
\in\mathcal{H}$. A \textit{purification} of a density operator $\rho
\in\mathcal{D}(\mathcal{H})$ is a pure extension of $\rho$.
Throughout this work, we restrict ourselves to finite-dimensional Hilbert
spaces, so that a $d$-dimensional Hilbert space is isomorphic to
$\mathbb{C}^{d}$. A quantum channel is a linear, completely positive,
trace-preserving (CPTP) map $\mathcal{N}:\mathcal{D}(\mathcal{H}_{\text{in}%
})\rightarrow\mathcal{D}(\mathcal{H}_{\text{out}})$. The Stinespring
representation theorem states that every CPTP\ map can be realized by
tensoring its input with an ancillary environment system in some fiducial
state $\left\vert 0\right\rangle _{E}\in\mathcal{H}_{E}$ where $\dim\left(
\mathcal{H}_{E}\right)  \leq\dim(\mathcal{H}_{\text{in}})\dim(\mathcal{H}%
_{\text{out}})$, performing some unitary operation on the joint Hilbert space
$\mathcal{H}_{\text{in}}\otimes\mathcal{H}_{E}$, factoring the unitary's
output Hilbert space as $\mathcal{H}_{\text{out}}\otimes\mathcal{H}%
_{E^{\prime}}$, and finally tracing over the Hilbert space $\mathcal{H}%
_{E^{\prime}}$ \cite{S55}. That is, for every CPTP\ map $\mathcal{N}$, there
exists some unitary $U$ such that the following relation holds for all
$\rho\in\mathcal{D}(\mathcal{H}_{\text{in}})$:%
\[
\mathcal{N}\left(  \rho\right)  =\text{Tr}_{E^{\prime}}\left\{  U\left(
\rho\otimes\left\vert 0\right\rangle \left\langle 0\right\vert _{E}\right)
U^{\dag}\right\}  .
\]
This theorem is the essential reason for the equivalence in computational
power between the unitary and mixed-state circuit models of quantum
computation \cite{Aharonov:1998:QCM:276698.276708}.

\subsection{Distance measures}

The trace norm of an operator $A$ is $\left\Vert A\right\Vert _{1}\equiv
\ $Tr$\{\sqrt{A^{\dag}A}\}$. The metric on quantum states induced by the trace
norm is called the trace distance, which has an operational interpretation as the
bias when using an optimal measurement to distinguish two quantum states
$\rho$ and $\sigma$ \cite{NC00}. That is, when $\rho$ and $\sigma$ are chosen
uniformly at random, the probability $p_{\text{succ}}$ of successfully
discriminating them with an optimal measurement is as follows:%
\[
p_{\text{succ}}=\frac{1}{2}\left(  1+\frac{1}{2}\left\Vert \rho-\sigma
\right\Vert _{1}\right)  .
\]
There is also a variational characterization of the trace distance as%
\[
\left\Vert \rho-\sigma\right\Vert _{1}=2\max_{0\leq\Lambda\leq I}%
\text{Tr}\left\{  \Lambda\left(  \rho-\sigma\right)  \right\}  ,
\]
which leads to the following useful inequality that holds for all $\Gamma$
such that $0\leq\Gamma\leq I$:%
\begin{equation}
\text{Tr}\left\{  \Gamma\rho\right\}  \geq\text{Tr}\left\{  \Gamma
\sigma\right\}  -\left\Vert \rho-\sigma\right\Vert _{1}%
.\label{eq:trace-inequality}%
\end{equation}

Suppose now that $\rho_{AB}$ and $\sigma_{AB}$ are in $\mathcal{D}\left(
\mathcal{H}_{A}\otimes\mathcal{H}_{B}\right)  $. Then the one-way local
operations and classical communication (1-LOCC) distance between these two
states, induced by a 1-LOCC norm \cite{MWW09}, is given\ by%
\[
\left\Vert \rho_{AB}-\sigma_{AB}\right\Vert _{1-\text{LOCC}}\equiv
\max_{\Lambda_{B\rightarrow X}}\left\Vert \left(  I_{A}\otimes\Lambda
_{B\rightarrow X}\right)  \left(  \rho_{AB}-\sigma_{AB}\right)  \right\Vert
_{1},
\]
where the maximization on the RHS\ is over all quantum-to-classical channels%
\[
\Lambda_{B\rightarrow X}\left(  \omega\right)  \equiv\sum_{x\in\mathcal{X}%
}\text{Tr}\left\{  \Lambda_{x}\omega\right\}  \left\vert x\right\rangle
\left\langle x\right\vert
\]
with $\Lambda_{x}\geq0$ for all $x\in\mathcal{X}$, $\sum_{x\in\mathcal{X}%
}\Lambda_{x}=I$, and $\left\{  \left\vert x\right\rangle \right\}  $ some
orthonormal basis. (Note that we could also define the 1-LOCC\ distance with
respect to measurement maps on Alice's system, which gives a value generally
different from the above.)\ The 1-LOCC distance is equal to the bias in
success probability when using an optimal one-way LOCC protocol to distinguish
the states:%
\[
p_{\text{succ,1-LOCC}}=\frac{1}{2}\left(  1+\frac{1}{2}\left\Vert \rho
_{AB}-\sigma_{AB}\right\Vert _{1-\text{LOCC}}\right)  .
\]
This distance is the natural distance measure in
the setting of Bell experiments \cite{bell1964}\ or quantum teleportation, for
example. Also, it follows that%
\begin{equation}
\left\Vert \rho-\sigma\right\Vert _{1-\text{LOCC}}\leq\left\Vert \rho
-\sigma\right\Vert _{1},\label{eq:difference-trace-LOCC}%
\end{equation}
because a general protocol to distinguish $\rho$ from $\sigma$ never performs
worse than one restricted to one-way LOCC\ operations.

The quantum fidelity $F\left(  \rho,\sigma\right)  $ between two quantum
states $\rho$ and $\sigma$ is another measure of distinguishability, defined
as follows:%
\begin{equation}
F\left(  \rho,\sigma\right)  \equiv\left\Vert \sqrt{\rho}\sqrt{\sigma
}\right\Vert _{1}^{2}.\label{eq:fidelity-norm}%
\end{equation}
Uhlmann characterized the fidelity as the optimal squared overlap between any
two purifications of $\rho$ and $\sigma$ \cite{U73}:%
\[
F\left(  \rho,\sigma\right)  =\max_{\left\vert \phi_{\rho}\right\rangle
,\left\vert \phi_{\sigma}\right\rangle }\left\vert \left\langle \phi_{\rho
}|\phi_{\sigma}\right\rangle \right\vert ^{2}.
\]
Thus, the fidelity has an operational interpretation as the optimal
probability with which a purification of $\rho$ would pass a test for being a
purification of $\sigma$. Since all purifications are related by a unitary
operation on the purifying system, this characterization is equivalent to the
following one:%
\begin{equation}
F\left(  \rho,\sigma\right)  =\max_{U}\left\vert \left\langle \phi_{\rho
}|\left(  U\otimes I_{\mathcal{H}}\right)  |\phi_{\sigma}\right\rangle
\right\vert ^{2},\label{eq:uhlmann-thm}%
\end{equation}
where $\left\vert \phi_{\rho}\right\rangle $ and $\left\vert \phi_{\sigma
}\right\rangle $ are now two fixed purifications of $\rho$ and $\sigma$,
respectively, and the optimization is over all unitaries acting on the
purifying system (the fact that (\ref{eq:fidelity-norm}) is equal to
(\ref{eq:uhlmann-thm}) is often referred to as Uhlmann's theorem). The
fidelity and trace distance are related by the Fuchs-van-de-Graaf
inequalities~\cite{FvG99}:%
\begin{equation}
1-\sqrt{F\left(  \rho,\sigma\right)  }\leq\frac{1}{2}\left\Vert \rho
-\sigma\right\Vert _{1}\leq\sqrt{1-F\left(  \rho,\sigma\right)  }%
.\label{eq:FvG-ineqs}%
\end{equation}

\subsection{Separability, $k$-extendibility, and the maximum $k$-extendible
fidelity}

A bipartite state $\rho_{AB}\in\mathcal{D}\left(  \mathcal{H}_{A}%
\otimes\mathcal{H}_{B}\right)  $ is $k$-extendible \cite{DPS02,DPS04}\ if
there exists a state $\omega_{AB_{1}\cdots B_{k}}\in\mathcal{D}\left(
\mathcal{H}_{A}\otimes\mathcal{H}_{B_{1}}\otimes\cdots\otimes\mathcal{H}%
_{B_{k}}\right)  $ such that

\begin{enumerate}
\item Each Hilbert space $\mathcal{H}_{B_{i}}$ is isomorphic to $\mathcal{H}%
_{B}$ for all $i\in\left\{  1,\ldots,k\right\}  $.

\item The state $\omega_{AB_{1}\cdots B_{k}}$ is invariant under permutations
of the systems $B_{1}$ through $B_{k}$. That is,%
\begin{equation}
\forall\pi\in S_{k}:\omega_{AB_{1}\cdots B_{k}}=\left(  I_{A}\otimes
W_{B_{1}\cdots B_{k}}^{\pi}\right)  \omega_{AB_{1}\cdots B_{k}}\left(
I_{A}\otimes W_{B_{1}\cdots B_{k}}^{\pi}\right)  ^{\dag},
\label{eq:perm-inv-cond}%
\end{equation}
where $S_{k}$ is the symmetric group on $k$ elements and $W_{B_{1}\cdots
B_{k}}^{\pi}$ is a unitary operation that implements the permutation $\pi$ of
the $B$ systems.

\item The state $\omega_{AB_{1}\cdots B_{k}}$ is an extension of $\rho_{AB}$:%
\[
\rho_{AB}=\text{Tr}_{B_{2}\cdots B_{k}}\left\{  \omega_{AB_{1}\cdots B_{k}%
}\right\}  .
\]

\end{enumerate}

Let $\mathcal{E}_{k}$ denote the set of $k$-extendible states. Every separable
state is $k$-extendible for all $k\geq2$. Since every separable state has a
decomposition of the form in (\ref{eq:separable-state-pure-decomp}), an
obvious choice for a $k$-extension is%
\begin{equation}
\sum_{x\in\mathcal{X}}p_{X}\left(  x\right)  \ \left\vert \psi_{x}%
\right\rangle \left\langle \psi_{x}\right\vert _{A}\otimes\left\vert \phi
_{x}\right\rangle \left\langle \phi_{x}\right\vert _{B_{1}}\otimes
\cdots\otimes\left\vert \phi_{x}\right\rangle \left\langle \phi_{x}\right\vert
_{B_{k}}.\label{eq:sep-k-extension}%
\end{equation}
On the other hand, if a state is not separable, there always exists some $k$
for which the state is not $k$-extendible, and furthermore, for every $l>k$,
the state is also not $l$-extendible \cite{DPS02,DPS04}. In this sense, the
set $\mathcal{E}_{k}$ forms an approximation to the set $\mathcal{S}$ of
separable states, and the approximation becomes exact in the limit as
$k\rightarrow\infty$.

The \textit{maximum }$k$\textit{-extendible fidelity} of a state $\rho_{AB}$
is defined as%
\[
\max_{\sigma_{AB}\in\mathcal{E}_{k}}F\left(  \rho_{AB},\sigma_{AB}\right)  .
\]
In this paper, we give the maximum $k$-extendible fidelity an operational
interpretation as an upper bound on the maximum probability with which a
prover can convince a verifier to accept in our \textsf{QIP(2)} protocol that
tests for $k$-extendibility. Clearly, the above quantity converges to the
\textit{maximum separable fidelity} (defined in \cite{W04}) in the limit as
$k\rightarrow\infty$:%
\[
\max_{\sigma_{AB}\in\mathcal{S}}F\left(  \rho_{AB},\sigma_{AB}\right)
=\lim_{k\rightarrow\infty}\max_{\sigma_{AB}\in\mathcal{E}_{k}}F\left(
\rho_{AB},\sigma_{AB}\right)  .
\]

Finally, we state a lemma which extends Theorem~3 of \cite{BCY11a}. This lemma
establishes a notion of approximate $k$-extendibility which is essential for
our work here. The proof of this lemma is a straightforward modification of
the proof of Theorem~3 of \cite{BCY11a} and we provide it in the appendix.

\begin{lemma}
\label{cor:contra-approx-k-ext}Let $\rho_{AB}$ be $\varepsilon$-away in
one-way LOCC distance from the set of separable states, for some
$\varepsilon>0$:%
\[
\min_{\sigma_{AB}\in\mathcal{S}}\left\Vert \rho_{AB}-\sigma_{AB}\right\Vert
_{\text{1-LOCC}}\geq\varepsilon.
\]
Then the state $\rho_{AB}$ is $\delta$-away in trace distance from the set of
$k$-extendible states:%
\[
\min_{\sigma_{AB}\in\mathcal{E}_{k}}\left\Vert \rho_{AB}-\sigma_{AB}%
\right\Vert _{1}\geq\delta,
\]
for $\delta < \varepsilon$ and where%
\[
k=\left\lceil \frac{16\ln2\ \log\left\vert A\right\vert }{(\varepsilon - \delta)
^{2}  }\right\rceil .
\]

\end{lemma}

\subsection{Quantum interactive proof systems}

Quantum interactive proof systems were first defined and analyzed in
\cite{W02,KW00}. Their study is an important component of quantum
computational complexity theory, and they are essential in our work here. A
quantum interactive proof system involves an exchange of quantum messages
between a polynomial-time quantum verifier and a computationally unbounded
quantum prover. Since any quantum interactive proof system can be parallelized
such that at most three messages are exchanged between the verifier and prover
(\textsf{QIP} $=$ \textsf{QIP(3)} up to a wide range of parameters)
\cite{KW00}, we focus our attention on quantum interactive proof systems with
three or fewer messages, starting by defining \textsf{QIP(3)}.

\subsubsection{Three-message quantum interactive proof systems}

Given an input string $x$ of length $n$, a quantum verifier consists of two
unitary quantum circuits $V_{1}$ and $V_{2}$\ computed from $x$ in polynomial
time. These circuits are generated from some finite gate set that is universal
for quantum computation \cite{NC00,W09}. The qubits of the verifier are
divided into two sets: private qubits and message qubits. The verifier
exchanges the message qubits with the prover, while keeping the private qubits
in his own laboratory. A quantum prover is defined similarly to the
verifier---he has private qubits and message qubits that he exchanges with the
verifier. However, he is allowed to perform two arbitrary, unconstrained
unitary quantum operations $P_{1}$ and $P_{2}$\ on the qubits in his
laboratory.%
\begin{figure}
[ptb]
\begin{center}
\includegraphics[
natheight=2.239900in,
natwidth=9.147100in,
width=6.5111in
]%
{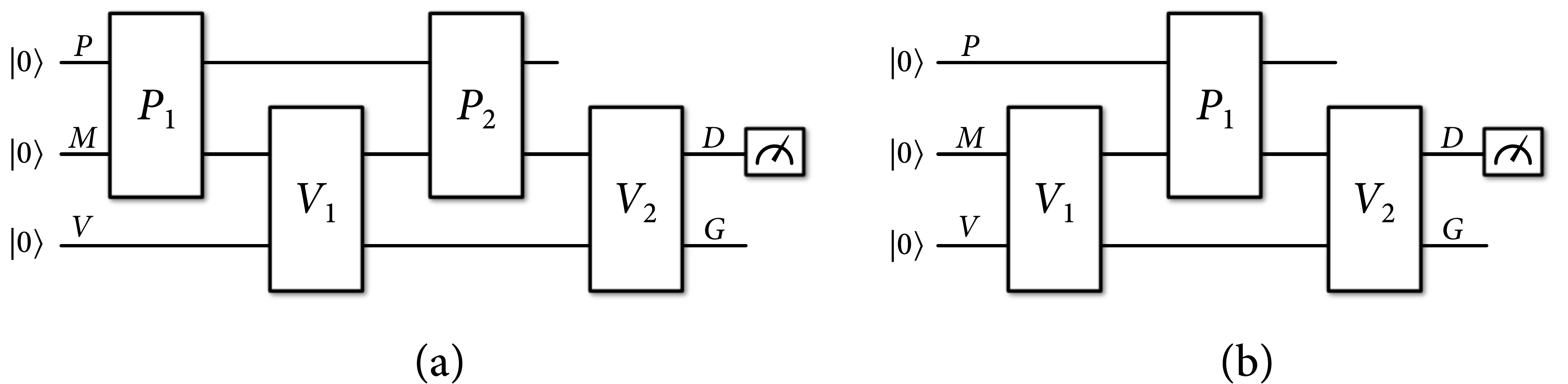}%
\caption{Quantum interactive proof systems with (a)\ three messages and
(b)\ two messages exchanged between the verifier and the prover.}%
\label{fig:qip-3-and-2}%
\end{center}
\end{figure}

Figure~\ref{fig:qip-3-and-2}(a) depicts a quantum interactive proof system
with three messages. $P$ is the register containing the prover's private
qubits, $M$ contains the message qubits, and $V$ contains the verifier's
private qubits. For simplicity, we take the convention that these registers
may change size throughout the execution of the protocol, as long as registers
$M$ and $V$ are polynomial in the length of the input $x$. The quantum
interactive proof system begins with the prover initializing the qubits in his
registers $P$ and $M$ to the all-zero state, and the verifier does the same to
his qubits in the register $V$. The prover then acts with the unitary $P_{1}$
on registers $P$ and $M$ and sends the message qubits to the verifier. They
proceed along the lines indicated in Figure~\ref{fig:qip-3-and-2}(a), until
the verifier acts with a final unitary $V_{2}$. This final unitary has two
output systems: a decision qubit in register$~D$ and other \textquotedblleft
garbage qubits\textquotedblright\ in register $G$. The verifier then measures
the decision qubit in the computational basis to decide whether to accept or
reject. The maximum probability with which the prover can make the verifier
accept is equal to%
\begin{equation}
\max_{\left\vert \psi\right\rangle _{PM},\ P_{2}}\left\Vert \left\langle
1\right\vert _{D}V_{2}P_{2}V_{1}\left\vert \psi\right\rangle _{PM}\left\vert
0\right\rangle _{V}\right\Vert _{2}^{2},\label{eq:max-accept-qip-3}%
\end{equation}
where the maximization is over all pure states $\left\vert \psi\right\rangle
_{PM}=P_{1}\left\vert 0\right\rangle _{PM}$ that the prover can prepare at the
beginning of the protocol and all unitaries $P_{2}$ that he can perform.

By defining $\mathcal{N}_{1}$ as the quantum channel induced by applying the
verifier's first unitary and tracing over the verifier's message register $M$:%
\[
\mathcal{N}_{1}\left(  \rho\right)  \equiv\text{Tr}_{M}\left\{  V_{1}\left(
\rho\otimes\left\vert 0\right\rangle \left\langle 0\right\vert _{V}\right)
V_{1}^{\dag}\right\}  ,
\]
defining $\mathcal{N}_{2}$ as the quantum channel induced by applying the
inverse of the verifier's final unitary and tracing over the verifier's
message register $M$:%
\begin{equation}
\mathcal{N}_{2}\left(  \sigma\right)  \equiv\text{Tr}_{M}\left\{  V_{2}^{\dag
}\left(  \sigma_{G}\otimes\left\vert 1\right\rangle \left\langle 1\right\vert
_{D}\right)  V_{2}\right\}  ,\label{eq:channel-from-final-unitary-verifier}%
\end{equation}
and applying Uhlmann's theorem in (\ref{eq:uhlmann-thm}) to the maximum
acceptance probability in (\ref{eq:max-accept-qip-3}), one can rewrite
(\ref{eq:max-accept-qip-3}) as the maximum output fidelity of the channels
$\mathcal{N}_{1}$ and $\mathcal{N}_{2}$ over all possible inputs to them
\cite{KW00,RW05,R09}:%
\begin{equation}
\max_{\left\vert \psi\right\rangle _{PM},\ P_{2}}\left\Vert \left\langle
1\right\vert _{D}V_{2}P_{2}V_{1}\left\vert \psi\right\rangle _{PM}\left\vert
0\right\rangle _{V}\right\Vert _{2}^{2}=\max_{\rho,\ \sigma}F\left(
\mathcal{N}_{1}\left(  \rho\right)  ,\mathcal{N}_{2}\left(  \sigma\right)
\right)  .\label{eq:max-output-fidelity}%
\end{equation}
(See Chapter~4 of \cite{R09} for details of this calculation).

We can now define the complexity class \textsf{QIP(3)}:

\begin{definition}
[\textsf{QIP(3)}]\label{def:QIP(3)}Let $A=(A_{\text{yes}},A_{\text{no}})$ be a
promise problem, and let $c,s:\mathbb{N}\rightarrow\left[  0,1\right]  $ be
polynomial-time computable functions such that the gap $c-s$ is at least an
inverse polynomial in the input size. Then $A\in\ $\textsf{QIP}$\left(
3,c,s\right)  $ if there exists a three-message quantum interactive proof
system with the following properties:

\begin{enumerate}
\item Completeness:\ For all input strings $x\in A_{\text{yes}}$, there exists
a prover that causes the verifier to accept with probability at least
$c\left(  \left\vert x\right\vert \right)  $.

\item Soundness:\ For all input strings $x\in A_{\text{no}}$, every prover
causes the verifier to accept with probability at most $s\left(  \left\vert
x\right\vert \right)  $.
\end{enumerate}
\end{definition}

Note that one can amplify the gap between $c$ and $s$ to a
constant by reducing the proof system to one with perfect completeness 
followed by employing parallel repetition \cite{KW00}.

Given the fact that we can rewrite the maximum acceptance probability of any
\textsf{QIP(3)} protocol as in (\ref{eq:max-output-fidelity}), it is
straightforward to formulate a complete promise problem for \textsf{QIP(3)}
called \textsf{CLOSE-IMAGES}, which essentially just amounts to rewriting the
above definition \cite{RW05,R09}:

\begin{problem}
[\textsf{CLOSE-IMAGES}]Fix two constants $c,s\in\left[  0,1\right]  $ such
that $c>s$. Given are two mixed-state quantum circuits $Q_{0}$ and $Q_{1}$,
each accepting $n$-qubit inputs and having $m$-qubit outputs. Decide whether

\begin{enumerate}
\item Yes:\ There exist $n$-qubit states $\rho_{0}$ and $\rho_{1}$ such that
$\max_{\rho_{0},\rho_{1}}F\left(  Q_{0}\left(  \rho_{0}\right)  ,Q_{1}\left(
\rho_{1}\right)  \right)  \geq c$.

\item No:\ For all $n$-qubit states $\rho_{0}$ and $\rho_{1}$, it holds that
$\max_{\rho_{0},\rho_{1}}F\left(  Q_{0}\left(  \rho_{0}\right)  ,Q_{1}\left(
\rho_{1}\right)  \right)  \leq s$.
\end{enumerate}
\end{problem}

The following promise problem is also complete for \textsf{QIP(3)}, but
proving so requires more than a trivial rewriting of the definition of
\textsf{QIP(3)} \cite{RW05,R09}:

\begin{problem}
[\textsf{QUANTUM-CIRCUIT-DISTINGUISHABILITY}]\label{prob:QCD}Fix a constant
$\varepsilon\in\lbrack0,1)$. Given are two mixed-state quantum circuits
$Q_{0}$ and $Q_{1}$, each with $n$-qubit inputs and $m$-qubit outputs. Decide whether

\begin{enumerate}
\item Yes:\ There is a quantum input for which the circuits are
distinguishable:%
\[
\max_{\rho\in\mathcal{D}(\mathcal{H}_{R}\otimes\mathcal{H}_{\text{in}}%
)}\left\Vert \left(  I_{R}\otimes Q_{0}\right)  \left(  \rho\right)  -\left(
I_{R}\otimes Q_{1}\right)  \left(  \rho\right)  \right\Vert _{1}%
\geq2-\varepsilon.
\]

\item No: No quantum input can distinguish the circuits:%
\[
\max_{\rho\in\mathcal{D}(\mathcal{H}_{R}\otimes\mathcal{H}_{\text{in}}%
)}\left\Vert \left(  I_{R}\otimes Q_{0}\right)  \left(  \rho\right)  -\left(
I_{R}\otimes Q_{1}\right)  \left(  \rho\right)  \right\Vert _{1}%
\leq\varepsilon.
\]

\end{enumerate}
\end{problem}

In what follows, we abbreviate \textsf{QUANTUM-CIRCUIT-DISTINGUISHABILITY}\ as
\textsf{QCD}.

\subsubsection{Two-message quantum interactive proof systems}

Now that we have defined three-message quantum interactive proof systems, it
is straightforward to define two-message quantum interactive proof systems.
The description is essentially identical to that given in the first two
paragraphs of the previous section, with the exception that the verifier acts
both first and last with unitaries $V_{1}$ and $V_{2}$, while the prover acts
between these two operations with a unitary $P_{1}$ (see
Figure~\ref{fig:qip-3-and-2}(a)). Thus, the maximum probability with which the
prover can make the verifier accept is equal to%
\[
\max_{P_{2}}\left\Vert \left\langle 1\right\vert _{D}V_{2}P_{2}\left\vert
0\right\rangle _{P}\left\vert \phi\right\rangle _{MV}\right\Vert _{2}^{2},
\]
where $\left\vert \phi\right\rangle _{MV}=V_{1}\left\vert 0\right\rangle
_{MV}$. By following essentially the same reasoning as before, we can rewrite
this maximum acceptance probability in terms of the quantum fidelity. First,
we define the mixed state%
\[
\omega_{V}=\text{Tr}_{M}\left\{  V_{1}\ \left\vert 0\right\rangle \left\langle
0\right\vert _{MV}\ V_{1}^{\dag}\right\}  .
\]
We also define the quantum channel $\mathcal{N}$ from $V_{2}$ as in
(\ref{eq:channel-from-final-unitary-verifier}). By again applying Uhlmann's
theorem, it is straightforward to prove that the maximum acceptance
probability is equal to the fidelity between the state $\omega$ and the
channel $\mathcal{N}$ when maximizing over all inputs to the channel:%
\begin{equation}
\max_{P_{2}}\left\Vert \left\langle 1\right\vert _{D}V_{2}P_{2}\left\vert
0\right\rangle _{P}\left\vert \phi\right\rangle _{MV}\right\Vert _{2}^{2}%
=\max_{\sigma}F\left(  \omega,\mathcal{N}\left(  \sigma\right)  \right)  .
\label{eq:fidelity-QIP(2)}%
\end{equation}
The definition of the complexity class \textsf{QIP(2)} is identical to that
given in Definition~\ref{def:QIP(3)} (but substituting \textsf{QIP(2)} for
\textsf{QIP(3)}). One can amplify the gap between $c$ and $s$ to a constant
by taking an AND of majorities as done in \cite{JUW09}.
Also, the following promise problem, called
\textsf{CLOSE-IMAGE}, is a complete promise problem for \textsf{QIP(2)}:

\begin{problem}
[\textsf{CLOSE-IMAGE}]Fix two constants $c,s\in\left[  0,1\right]  $ such that
$c>s$. Given is a mixed-state quantum circuit to generate the $m$-qubit state
$\rho_{0}$ and a mixed-state quantum circuit $Q_{1}$, with an $n$-qubit input
state and an $m$-qubit output state. Decide whether

\begin{enumerate}
\item Yes:\ There exists an $n$-qubit state $\rho_{1}$ such that $\max
_{\rho_{1}}F\left(  \rho_{0},Q_{1}\left(  \rho_{1}\right)  \right)  \geq c$.

\item No:\ For all $n$-qubit states $\rho_{1}$, it holds that $\max_{\rho_{1}%
}F\left(  \rho_{0},Q_{1}\left(  \rho_{1}\right)  \right)  \leq s$.
\end{enumerate}
\end{problem}

The fact that \textsf{CLOSE-IMAGE}\ is complete for \textsf{QIP(2)} follows
essentially the same reasoning as in \cite{R09} (it amounts to a rewriting of
the definition of \textsf{QIP(2)}).

\subsubsection{Quantum statistical zero-knowledge proof systems}

Another kind of quantum interactive proof system that is relevant for our work
here is a quantum statistical zero-knowledge (\textsf{QSZK}) proof system
\cite{W02,W09zkqa}. The definition of the \textquotedblleft
honest-verifier\textquotedblright\ version of this complexity class is
essentially the same as that for \textsf{QIP} (with an arbitrary number of
messages exchanged), but the difference is that in the case of a positive
problem instance, the states of the verifier before and after every
interaction with the prover should be such that he could have actually
generated them himself. In this sense, he does not gain any \textquotedblleft
knowledge\textquotedblright\ by interacting with the prover (other than being
convinced to accept). Watrous has shown that any honest-verifier QSZK\ proof
system has an equivalent proof system in which the verifier is not required to
behave honestly \cite{W09zkqa}.

The following promise problem is complete for \textsf{QSZK} \cite{W02}:

\begin{problem}
[\textsf{QUANTUM-STATE-DISTINGUISHABILITY}]Fix a constant $\varepsilon
\in\lbrack0,1)$. Given is a mixed-state quantum circuit to generate the
$n$-qubit states $\rho_{0}$ and $\rho_{1}$. Decide whether

\begin{enumerate}
\item Yes:\ $\left\Vert \rho_{0}-\rho_{1}\right\Vert _{1}\geq2-\varepsilon$.

\item No:\ $\left\Vert \rho_{0}-\rho_{1}\right\Vert _{1}\leq\varepsilon$.
\end{enumerate}
\end{problem}

In what follows, we abbreviate \textsf{QUANTUM-STATE-DISTINGUISHABILITY}\ as
\textsf{QSD}.

\section{A two-message quantum interactive proof system decides the quantum
separability problem}

\label{sec:qsep-in-qip2}We are now ready to provide a formal statement of the
promise problem \textsf{QSEP-STATE}$_{\operatorname{1,1-LOCC}}$, and we follow with a proof that it is
decidable by a two-message quantum interactive proof system.

\begin{problem}
[\textsf{QSEP-STATE}$_{\operatorname{1,1-LOCC}}$$\left(  \delta_{c},\delta_{s}\right)  $]%
\label{prob:QSEP-CIRCUIT}Given is a mixed-state quantum circuit to generate
the $n$-qubit state $\rho_{AB}$, along with a labeling of the qubits in the
reference system $R$ and the output qubits for $A$ and $B$. Decide whether

\begin{enumerate}
\item Yes:\ There is a separable state $\sigma_{AB}\in\mathcal{S}$\ that is $\delta_{c}$-close to $\rho_{AB}$ in trace distance:%
\[
\min_{\sigma_{AB}\in\mathcal{S}}\left\Vert \rho_{AB}-\sigma_{AB}\right\Vert
_{1}\leq\delta_{c}.
\]

\item No:\ Every separable state is at least $\delta_{s}%
$-far from $\rho_{AB}$  in 1-LOCC$\ $distance:%
\[
\min_{\sigma_{AB}\in\mathcal{S}}\left\Vert \rho_{AB}-\sigma_{AB}\right\Vert
_{1-\text{LOCC}} \geq \delta_{s}.
\]

\end{enumerate}
\end{problem}

\begin{theorem}
\label{thm:QSEP-in-QIP(2)}\textsf{QSEP-STATE}$_{\operatorname{1,1-LOCC}}$$\left(  \delta_{c},\delta
_{s}\right)  $\textsf{\ }$\in\ $\textsf{QIP(2)} if there are polynomial-time
computable functions $\delta_{c},\delta_{s}:\mathbb{N}\rightarrow\left[
0,1\right]  $, such that the difference $\delta_{s}^2 / 8  -2\sqrt{\delta_{c}}$
is larger than an inverse polynomial in the circuit size.\end{theorem}

\begin{proof}
Figure~\ref{fig:qsep-qip2}\ depicts a two-message quantum interactive proof
system for \textsf{QSEP-STATE}$_{\operatorname{1,1-LOCC}}$. The protocol begins with the verifier
preparing the state $\left\vert \psi_{\rho}\right\rangle _{RAB}$, a particular
purification of $\rho_{AB}$, by running the quantum circuit $U_{\rho}%
$\ specified by the input string (the problem instance). The verifier
transmits the reference system to the prover, who then acts on $R$ and some
ancillary qubits with a unitary $P_{1}$\ that has output systems $R^{\prime}$,
$B_{2}$, \ldots, $B_{k}$. The prover transmits systems $B_{2}$, \ldots,
$B_{k}$ to the verifier. The verifier then performs phase estimation over the
symmetric group \cite{K95,BBDEJM97} (also known as the \textquotedblleft
permutation test\textquotedblright\ \cite{KNY08}) on the registers $B$,
$B_{2}$, \ldots, $B_{k}$, using the qubits in system $C$ as the control. The
verifier performs a computational basis measurement on all of the qubits in
the control register $C$ and accepts if and only if the measurement outcomes
are all zeros.%
\begin{figure}
[ptb]
\begin{center}
\includegraphics[
natheight=2.040100in,
natwidth=4.427000in,
width=4.478in
]%
{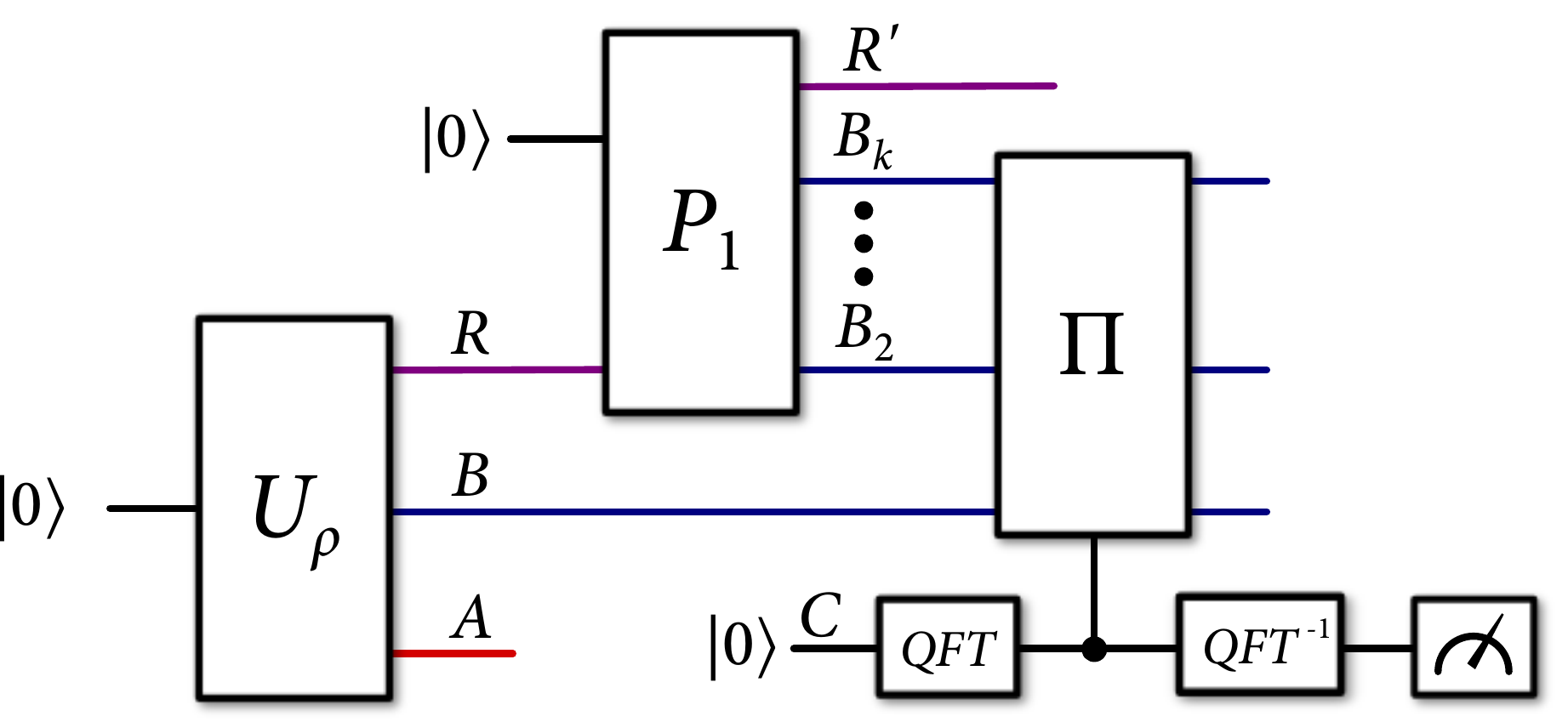}%
\caption{A two-message quantum interactive proof system for
\textsf{QSEP-STATE}$_{\operatorname{1,1-LOCC}}$. It begins with the verifier executing the circuit
$U_{\rho}$ that generates the state $\rho_{AB}$. He sends the reference system
to the prover. In the case that $\rho_{AB}$ is separable, the prover should be
able to act with a unitary on the reference system and some ancillas in order
to generate a purification of a $k$-extension of $\rho_{AB}$. The prover sends
all of the extension systems back to the verifier, who then performs phase
estimation over the symmetric group (a quantum Fourier transform followed by a
controlled permutation) in order to test if the state sent by the prover is
really a $k$-extension.}%
\label{fig:qsep-qip2}%
\end{center}
\end{figure}

This protocol is just an implementation of a $k$-extendibility test on a
quantum computer. We can build intuition for why it works on YES instances
by examining the exact case, when $\rho_{AB}$ is actually a separable state.
In this case, we know that $\rho_{AB}$ has a decomposition of the form given
in (\ref{eq:separable-state-pure-decomp}), and as such, it has an extension of
the form in (\ref{eq:sep-k-extension}) for all $k>1$. Thus, the following state is
a purification of $\rho_{AB}$:%
\[
\left\vert \phi_{k,\rho}\right\rangle _{R^{\prime}ABB_{2}\cdots B_{k}}%
\equiv\sum_{x\in\mathcal{X}}\sqrt{p_{X}\left(  x\right)  }\ \left\vert
x\right\rangle _{R^{\prime}}\otimes\left\vert \psi_{x}\right\rangle
_{A}\otimes\left\vert \phi_{x}\right\rangle _{B}\otimes\left\vert \phi
_{x}\right\rangle _{B_{2}}\otimes\cdots\otimes\left\vert \phi_{x}\right\rangle
_{B_{k}},
\]
where $\{\left\vert x\right\rangle _{R^{\prime}}\}$ is some orthonormal basis
for the reference system. Since all purifications are related by unitaries on
the reference system, the prover can append ancilla qubits to the $R$ system
received from the verifier and perform a unitary $P_{1}$ that takes
$\left\vert \psi_{\rho}\right\rangle _{RAB}\left\vert 0\right\rangle $ to
$\left\vert \phi_{k,\rho}\right\rangle _{R^{\prime}ABB_{2}\cdots B_{k}}$. The
prover then sends the systems $B_{2}$, \ldots, $B_{k}$ to the verifier. The
verifier performs a permutation test on the systems $B$, $B_{2}$, \ldots,
$B_{k}$. Since the state $\left\vert \phi_{k,\rho}\right\rangle _{R^{\prime
}ABB_{2}\cdots B_{k}}$ is invariant under permutations of the systems $B$,
$B_{2}$, \ldots, $B_{k}$, the qubits in the control register $C$\ do not
acquire a phase. Thus, after the final quantum Fourier transform is applied, the qubits
in the control register $C$ are in the all-zero state with certainty.

The analysis for a YES instance follows the above intuition closely. In this
case, there is some state $\sigma_{AB}\in\mathcal{S}$ that is $\delta_{c}%
$-close in trace distance to $\rho_{AB}$. By Uhlmann's theorem in
(\ref{eq:uhlmann-thm}) and the Fuchs-van-de-Graaf inequalities in
(\ref{eq:FvG-ineqs}), there is a purification $\left\vert \psi_{\sigma
}\right\rangle _{RAB}$\ of $\sigma_{AB}$ such that%
\begin{equation}
\left\Vert \left\vert \psi_{\rho}\right\rangle \left\langle \psi_{\rho
}\right\vert _{RAB}-\left\vert \psi_{\sigma}\right\rangle \left\langle
\psi_{\sigma}\right\vert _{RAB}\right\Vert _{1}\leq2\sqrt{\delta_{c}%
}.\label{eq:distance-pure-to-sep}%
\end{equation}
So the prover can just operate as above, but choosing his unitary $P_{1}$\ to
correspond to the state $\left\vert \psi_{\sigma}\right\rangle _{RAB}$
instead. Writing as $U$ the unitary corresponding to $P_{1}$ followed by the
permutation test, we obtain the following lower bound on the probability with
which the verifier accepts:%
\begin{align}
&  \text{Tr}\left\{  \left\vert 0\right\rangle \left\langle 0\right\vert
_{C}U\left(  \left\vert \psi_{\rho}\right\rangle \left\langle \psi_{\rho
}\right\vert _{RAB}\right)  U^{\dag}\right\}  \nonumber\\
&  =\text{Tr}\left\{  U^{\dag}\left\vert 0\right\rangle \left\langle
0\right\vert _{C}U\left(  \left\vert \psi_{\rho}\right\rangle \left\langle
\psi_{\rho}\right\vert _{RAB}\right)  \right\}  \nonumber\\
&  \geq\text{Tr}\left\{  U^{\dag}\left\vert 0\right\rangle \left\langle
0\right\vert _{C}U\left(  \left\vert \psi_{\sigma}\right\rangle \left\langle
\psi_{\sigma}\right\vert _{RAB}\right)  \right\}  -\left\Vert \left\vert
\psi_{\rho}\right\rangle \left\langle \psi_{\rho}\right\vert _{RAB}-\left\vert
\psi_{\sigma}\right\rangle \left\langle \psi_{\sigma}\right\vert
_{RAB}\right\Vert _{1} \nonumber\\
&  \geq1-2\sqrt{\delta_{c}}, \label{eq:QSEP-YES-analysis}
\end{align}
where the first inequality follows from (\ref{eq:trace-inequality}), and the
second inequality follows by applying (\ref{eq:distance-pure-to-sep}) and
because the protocol accepts with probability one for a separable state.

The analysis for a NO instance has two components:

\begin{enumerate}
\item demonstrating that the maximum $k$-extendible fidelity is an upper bound
on the maximum acceptance probability

\item using Lemma~\ref{cor:contra-approx-k-ext}\ regarding approximate
$k$-extendibility and the first item above to specify how large $k$ should be
in order to obtain a good upper bound on the maximum acceptance
probability.\footnote{For a YES instance, the value of $k$ does not matter
because the lower bound on the maximum acceptance probability is always as
given above.}
\end{enumerate}

We now discuss the first item above. Recall from (\ref{eq:fidelity-QIP(2)})
that the maximum acceptance probability of any QIP(2) system is equal to the
maximum fidelity between the state generated by the verifier's first circuit
and the channel generated by the inverse of the verifier's second circuit. For
the protocol in Figure~\ref{fig:qsep-qip2}, the state generated by the
verifier's first circuit is as follows:%
\[
\rho_{AB}\otimes\left\vert \text{perm}\right\rangle \left\langle
\text{perm}\right\vert _{C},
\]
where $\left\vert \text{perm}\right\rangle _{C}$ is a superposition over all
possible permutations of $k$ elements resulting from an application of the
quantum Fourier transform \cite{NC00}\ to the state $\left\vert 0\right\rangle
_{C}$:%
\begin{equation}
\left\vert \text{perm}\right\rangle _{C}\equiv\frac{1}{\sqrt{k!}}\sum_{\pi\in
S_{k}}\left\vert \pi\right\rangle _{C}, \label{eq:perm-state}
\end{equation}
so that the $C$ register requires $\left\lceil \log_{2}\left(  k!\right)
\right\rceil $ qubits. (Note that Figure~\ref{fig:qsep-qip2} depicts the
verifier generating $\left\vert \text{perm}\right\rangle _{C}$ later in the
protocol, but we could just as easily reorder things so that he generates this
state in the first step.) The channel generated by the inverse of the
verifier's circuit conditional on accepting is%
\[
\mathcal{M}_{ABB_{2}\cdots B_{k}\rightarrow ABC}\left(  \sigma_{ABB_{2}\cdots
B_{k}}\right)  \equiv\text{Tr}_{B_{2}\cdots B_{k}}\left\{  \left(  U_{\Pi
}\right)  _{BB_{2}\cdots B_{k}C}\left(  \sigma_{ABB_{2}\cdots B_{k}}%
\otimes\left\vert \text{perm}\right\rangle \left\langle \text{perm}\right\vert
_{C}\right)  \left(  U_{\Pi}^{\dag}\right)  _{BB_{2}\cdots B_{k}C}\right\}  ,
\]
where $\left(  U_{\Pi}\right)  _{BB_{2}\cdots B_{k}C}$ is a
controlled-permutation operation:%
\begin{equation}
\left(  U_{\Pi}\right)  _{BB_{2}\cdots B_{k}C}\equiv\sum_{\pi\in S_{k}%
}W_{BB_{2}\cdots B_{k}}^{\pi}\otimes\left\vert \pi\right\rangle \left\langle
\pi\right\vert _{C}, \label{eq:controlled-perm}
\end{equation}
and $W_{BB_{2}\cdots B_{k}}^{\pi}$ is a unitary operation corresponding to
permutation $\pi$ (mentioned before in (\ref{eq:perm-inv-cond})). So the
maximum acceptance probability is equal to%
\[
\max_{\sigma_{ABB_{2}\cdots B_{k}}}F\left(  \rho_{AB}\otimes\left\vert
\text{perm}\right\rangle \left\langle \text{perm}\right\vert _{C}%
,\mathcal{M}_{ABB_{2}\cdots B_{k}\rightarrow ABC}\left(  \sigma_{ABB_{2}\cdots
B_{k}}\right)  \right)  .
\]
Since the fidelity can only increase under the discarding of the control
register $C$,\footnote{We can interpret discarding the control register as
actually giving it to the prover, so that the resulting fidelity corresponds
to the maximum acceptance probability in a modified protocol in which the
prover controls the inputs to $C$.} the maximum acceptance probability is
upper bounded by the following quantity:%
\begin{equation}
\max_{\sigma_{ABB_{2}\cdots B_{k}}}F\left(  \rho_{AB},\mathcal{M}%
_{ABB_{2}\cdots B_{k}\rightarrow AB}\left(  \sigma_{ABB_{2}\cdots B_{k}%
}\right)  \right)  ,\label{eq:max-k-ext-fidelity-1}%
\end{equation}
where%
\begin{align*}
\mathcal{M}_{ABB_{2}\cdots B_{k}\rightarrow AB}\left(  \sigma_{ABB_{2}\cdots
B_{k}}\right)   &  =\text{Tr}_{C}\left\{  \mathcal{M}_{ABB_{2}\cdots
B_{k}\rightarrow ABC}\left(  \sigma_{ABB_{2}\cdots B_{k}}\right)  \right\}  \\
&  =\frac{1}{k!}\sum_{\pi\in S_{k}}\text{Tr}_{B_{2}\cdots B_{k}}\left\{
\left(  I_{A}\otimes W_{BB_{2}\cdots B_{k}}^{\pi}\right)  \sigma
_{ABB_{2}\cdots B_{k}}\left(  I_{A}\otimes W_{BB_{2}\cdots B_{k}}^{\pi
}\right)  ^{\dag}\right\}  ,
\end{align*}
which is just the channel that applies a random permutation of the $B$ systems
and discards the last $k-1$ systems $B_{2}$, \ldots, $B_{k}$. Clearly, since
the channel $\mathcal{M}_{ABB_{2}\cdots B_{k}\rightarrow AB}$ symmetrizes the
state of the systems $BB_{2}\cdots B_{k}$, the maximum in
(\ref{eq:max-k-ext-fidelity-1}) is achieved by a state $\sigma_{ABB_{2}\cdots
B_{k}}$ for which systems $BB_{2}\cdots B_{k}$ are permutation symmetric.
Thus, by recalling the definition of $k$-extendibility, we can rewrite
(\ref{eq:max-k-ext-fidelity-1}) as the maximum $k$-extendible fidelity of
$\rho_{AB}$:%
\begin{equation}
\max_{\sigma_{ABB_{2}\cdots B_{k}}}F\left(  \rho_{AB},\mathcal{M}%
_{ABB_{2}\cdots B_{k}\rightarrow AB}\left(  \sigma_{ABB_{2}\cdots B_{k}%
}\right)  \right)  =\max_{\sigma_{AB}\in\mathcal{E}_{k}}F\left(  \rho
_{AB},\sigma_{AB}\right)  .\label{eq:max-exists-arg}%
\end{equation}
This demonstrates that the maximum $k$-extendible fidelity is an upper bound
on the maximum acceptance probability and completes our proof of the first
item above.

The second part of the analysis of a NO instance involves determining how large $k$ needs to be. Suppose that%
\[
\min_{\sigma_{AB}\in\mathcal{S}}\left\Vert \rho_{AB}-\sigma_{AB}\right\Vert
_{\text{1-LOCC}}\geq\delta_{s}.
\]
According to Lemma~\ref{cor:contra-approx-k-ext}, if we take $k$ to be larger
than%
\[\left\lceil \frac{16\ln2\ \log\left\vert A\right\vert }{(\delta_s - \delta_s^{\prime})
^{2}  }\right\rceil .
\]
then we can guarantee that%
\[
\min_{\sigma_{AB}\in\mathcal{E}_{k}}\left\Vert \rho_{AB}-\sigma_{AB}%
\right\Vert _{1}\geq\delta_{s}^{\prime}.
\]
for $\delta_{s}^{\prime}$ strictly less than $\delta_s$.
We can enforce this latter condition by setting
$
\delta_{s}^{\prime}= \delta_{s} / \sqrt{2} .
$
Observe that $k$ is polynomial in
$n_{A}$, where $n_A$ is the number of qubits in Alice's system.
Then using the following manipulation of the Fuchs-van-de-Graaf
inequalities in (\ref{eq:FvG-ineqs}):
\[
F\left(  \rho,\sigma\right)  \leq1-\frac{1}{4}\left\Vert \rho-\sigma
\right\Vert _{1}^{2},
\]
we have that
\begin{align}
\max_{\sigma_{AB}\in\mathcal{E}_{k}}F\left(  \rho_{AB},\sigma_{AB}\right)   &
\leq1-\frac{1}{4}\min_{\sigma_{AB}\in\mathcal{E}_{k}}\left\Vert \rho
_{AB}-\sigma_{AB}\right\Vert _{1}^{2} \label{eq:bipartite-NO-analysis-1} \\
&  \leq1-\frac{1}{4}\left(  \delta_{s}^{\prime}\right)  ^{2} \\
&  = 1- \delta_{s}^2 / 8 \label{eq:bipartite-NO-analysis-3}
\end{align}

In the above, we have separated the probability of accepting and the
probability of rejecting by an inverse polynomial in $n_{A}$ (namely,
from the promise that the difference
$\delta_{s}^2 / 8  -2\sqrt{\delta_{c}}$ is at least an inverse polynomial in the
circuit size), and it is known that an inverse polynomial gap is sufficient to
place this protocol in \textsf{QIP(2)} (see Section~3.2 of Ref.~\cite{JUW09}
for how to amplify an inverse polynomial gap). Thus, we have given a
two-message quantum interactive proof system that decides the quantum
separability problem.
\end{proof}

\begin{remark}If one considers the characterization of \textsf{QIP(2)} in terms of the
complete promise problem \textsf{CLOSE-IMAGE}, one has to identify a state and
a channel to compare by means of the fidelity maximized over all inputs to the
channel. The natural state to consider for \textsf{QSEP-STATE}$_{\operatorname{1,1-LOCC}}$\ is
$\rho_{AB}$, while\ the natural channel to test for $k$-extendibility is one
that applies a random permutation to the systems $B$, $B_{2}$, \ldots, $B_{k}$
of a state $\sigma_{ABB_{2}\cdots B_{k}}$ and traces out the systems $B_{2}$,
\ldots, $B_{k}$. The Stinespring dilation of the channel is the circuit for
phase estimation over the symmetric group, so that in this sense, we can say
that \textsf{CLOSE-IMAGE}\ finds the phase estimation circuit given in
Figure~\ref{fig:qsep-qip2}.\end{remark}

\begin{remark}From the above proof, it is clear that if the promise regarding \textsf{QSEP-STATE}$_{\operatorname{1,1-LOCC}}$
is given in terms of fidelities
rather than the trace distance and the 1-LOCC distance, then the promise concerning the gap between $\delta_s$ and
$\delta_c$ could be that their difference is larger than an inverse polynomial (rather than making a promise
about the difference $\delta_{s}^2 / 8  -2\sqrt{\delta_{c}}$).\end{remark}

\begin{remark}The above two-message quantum interactive proof system also
solves a more traditional formulation of the quantum separability problem,
where the promise is that either 1) $\rho_{AB} \in \mathcal{S}$
or 2) $ \min_{\sigma_{AB} \in \mathcal{S}}
\Vert \rho_{AB} - \sigma_{AB} \Vert_2 \geq \delta_s$, where $\delta_s$ is
larger than an inverse polynomial in the circuit size. This follows from applying
the bound in (\ref{eq:matthews-bound}), due to Matthews et al.~\cite{MWW09}.
\end{remark}


\section{\textsf{QSZK}-hardness of the quantum separability problem}

\label{sec:qsep-qszk-hard}Having placed an upper bound on the difficulty of solving \textsf{QSEP-STATE}$_{\operatorname{1,1-LOCC}}$, we now move on to lower bounds, beginning in this section 
with a proof that it is hard for
\textsf{QSZK}. Our approach is to exhibit a Karp reduction from the
\textsf{QSZK}-complete promise problem \textsf{QSD}\ to \textsf{QSEP-STATE}$_{\operatorname{1,1-LOCC}}$.
The essential idea behind the reduction is similar to Rosgen and Watrous's
reduction of \textsf{CLOSE-IMAGES} to \textsf{QCD} \cite{RW05,R09}.

In order to demonstrate this reduction, we have to show that there is a
polynomial-time algorithm that encodes YES\ instances of \textsf{QSD}\ into
YES\ instances of \textsf{QSEP-STATE}$_{\operatorname{1,1-LOCC}}$ and the same for the NO instances.
Recall that for \textsf{QSD}, we are given a description of circuits
$U_{\rho_{0}}$ and $U_{\rho_{1}}$\ that generate mixed states $\rho_{0}$ and
$\rho_{1}$. The output qubits of the circuit are divided into two
sets:\ qubits in a reference system $R$\ that are traced over and qubits in a
system$~S$ which contains $\rho_{i}$. For $i\in\left\{  0,1\right\}  $, let%
\[
\left\vert \psi_{\rho_{i}}\right\rangle _{RS}\equiv U_{\rho_{i}}\left\vert
0\right\rangle ,
\]
so that%
\[
\rho_{i}=\text{Tr}_{R}\left\{  \left\vert \psi_{\rho_{i}}\right\rangle
\left\langle \psi_{\rho_{i}}\right\vert _{RS}\right\}  .
\]%
\begin{figure}
[ptb]
\begin{center}
\includegraphics[
natheight=2.412800in,
natwidth=2.239900in,
width=2.2796in
]%
{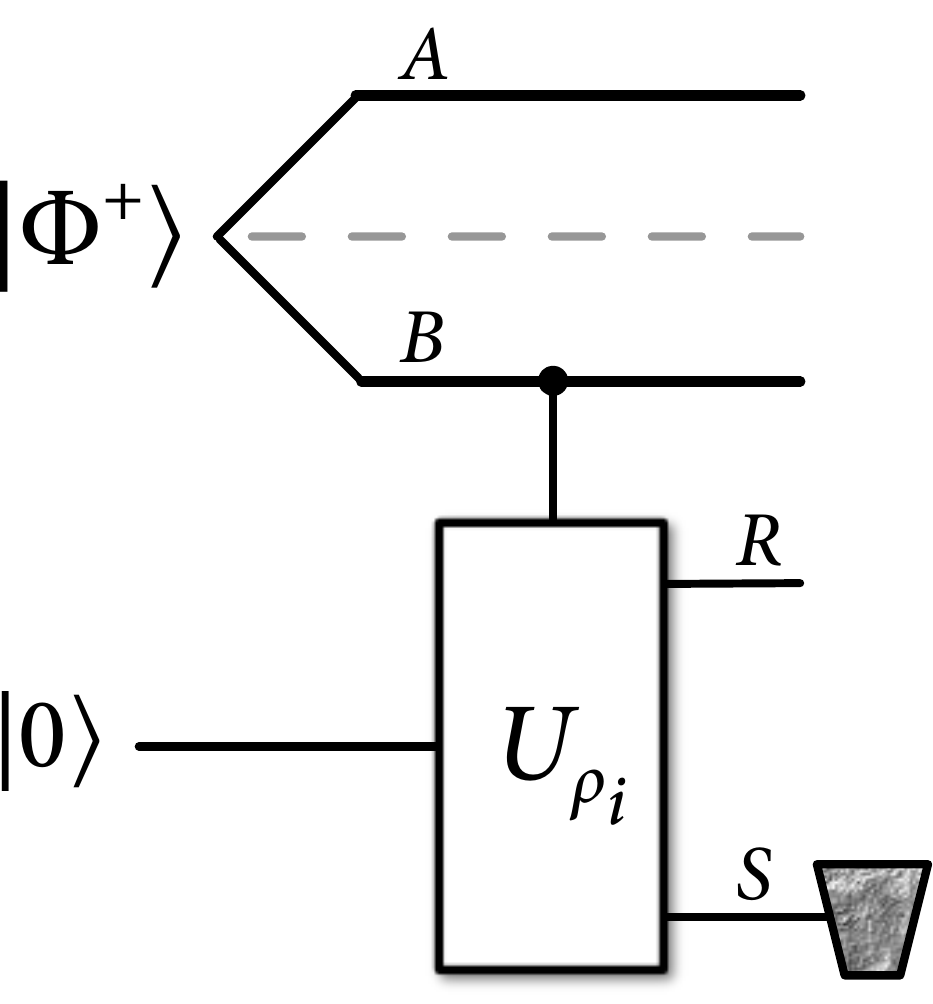}%
\caption{Given respective circuit descriptions $U_{\rho_{0}}$ and $U_{\rho
_{1}}$ for generating the states $\rho_{0}$ and $\rho_{1}$ on the output
system $S$, one can compute a description for the above circuit in polynomial
time, and furthermore, the above circuit can be run efficiently on a quantum
computer. This serves as a reduction from
\textsf{QUANTUM-STATE-DISTINGUISHABILITY} to \textsf{QSEP-STATE}$_{\operatorname{1,1-LOCC}}$, i.e.,
where one should decide if the state on systems $A$ and $BR$ is separable with
respect to this cut.}%
\label{fig:QSD-to-QSEP}%
\end{center}
\end{figure}

Figure~\ref{fig:QSD-to-QSEP}\ depicts a circuit that accomplishes the
reduction. From the description of the circuits $U_{\rho_{0}}$ and
$U_{\rho_{1}}$, one can generate a description of the circuit in
Figure~\ref{fig:QSD-to-QSEP}\ in polynomial time, and furthermore, the
resulting circuit runs efficiently on a quantum computer \cite{R09}. The
circuit takes as input a Bell state%
\[
\left\vert \Phi^{+}\right\rangle _{AB}\equiv\frac{1}{\sqrt{2}}\left(
\left\vert 00\right\rangle _{AB}+\left\vert 11\right\rangle _{AB}\right)  ,
\]
and performs the following controlled unitary from the qubit $B$ to the
ancilla qubits:%
\[
\left\vert 0\right\rangle \left\langle 0\right\vert _{B}\otimes U_{\rho_{0}%
}+\left\vert 1\right\rangle \left\langle 1\right\vert _{B}\otimes U_{\rho_{1}%
}.
\]
The resulting state is as follows:%
\[
\left\vert \varphi\right\rangle _{ABRS}\equiv\frac{1}{\sqrt{2}}\left(
\left\vert 0\right\rangle _{A}\left\vert 0\right\rangle _{B}\left\vert
\psi_{\rho_{0}}\right\rangle _{RS}+\left\vert 1\right\rangle _{A}\left\vert
1\right\rangle _{B}\left\vert \psi_{\rho_{1}}\right\rangle _{RS}\right)  .
\]
The output qubits are divided into three sets:\ environment qubits in the
system $S$ that are traced over, a single qubit in system $A$, and qubits in
systems $BR$. Thus, the state resulting from applying the circuit in
Figure~\ref{fig:QSD-to-QSEP}\ is as follows:%
\begin{equation}
\omega_{A:BR}\equiv\text{Tr}_{S}\left\{  \left\vert \varphi\right\rangle
\left\langle \varphi\right\vert _{ABRS}\right\}
.\label{eq:state-generated-by-reduction}%
\end{equation}
The task is to decide whether the state on systems $A$\ and $BR$ is separable
across this cut, subject to the promise in Problem~\ref{prob:QSEP-CIRCUIT}.
Our claim is that YES\ instances of \textsf{QSD}\ map to YES\ instances of
\textsf{QSEP-STATE}$_{\operatorname{1,1-LOCC}}$, with the same holding true for NO\ instances.

The intuition for why this reduction works is as follows. In the case of a
YES\ instance of \textsf{QSD}, the states $\rho_{0}$ and $\rho_{1}$ are
approximately orthogonal, so that tracing out the $S$ system of the circuit in
Figure~\ref{fig:QSD-to-QSEP}\ decoheres the Bell state, leaving a state on $A$
and $BR$ close to the following state:%
\begin{equation}
\omega_{A:BR}^{\text{sep}}\equiv\frac{1}{2}\left(  \left\vert 0\right\rangle
\left\langle 0\right\vert _{A}\otimes\left\vert 0\right\rangle \left\langle
0\right\vert _{B}\otimes\left(  \psi_{\rho_{0}}\right)  _{R}+\left\vert
1\right\rangle \left\langle 1\right\vert _{A}\otimes\left\vert 1\right\rangle
\left\langle 1\right\vert _{B}\otimes\left(  \psi_{\rho_{1}}\right)
_{R}\right)  . \label{eq:reduction-sep-state}%
\end{equation}
The above state is clearly separable with respect to the bipartite cut $A:BR$. In the case of a NO\ instance of \textsf{QSD}, the states $\rho_{0}$ and
$\rho_{1}$ are approximately indistinguishable, and tracing over the
$S$\ system of the circuit in Figure~\ref{fig:QSD-to-QSEP}\ does little to
decohere the entanglement shared between $A$ and $BR$. Thus, Bob can perform a
local unitary operation on systems $B$ and $R$ to distill out a pure Bell
state shared between $A$ and $B$. After this, Alice and Bob can perform a Bell
experiment on the distilled Bell state to determine if they indeed share a
Bell state. Since these two operations can be performed with local operations
and one message of classical communication, the resulting state is 1-LOCC
distinguishable from the set of separable states.

We now give a formal proof to justify this reduction:

\begin{theorem}
\label{thm:QSEP-QSZK-hard}\textsf{QSEP-STATE}$_{\operatorname{1,1-LOCC}}$ with constant promise gap is \textsf{QSZK}-hard.
\end{theorem}

\begin{proof}
We first prove that the circuit in Figure~\ref{fig:QSD-to-QSEP} maps YES
instances of \textsf{QSD} to YES\ instances of \textsf{QSEP-STATE}$_{\operatorname{1,1-LOCC}}$. So we
begin by assuming that%
\begin{equation}
\left\Vert \rho_{0}-\rho_{1}\right\Vert _{1}\geq2-\varepsilon
,\label{eq:co-QSD-YES}%
\end{equation}
and we will use this condition to show that the fidelity between
$\omega_{A:BR}^{\text{sep}}$ in (\ref{eq:reduction-sep-state}) and the reduced
state $\omega_{A:BR}$ in (\ref{eq:state-generated-by-reduction}) is close to
one. So, recall from Uhlmann's theorem in (\ref{eq:uhlmann-thm}) that the
fidelity between $\omega_{A:BR}^{\text{sep}}$ and $\omega_{A:BR}$ is the
maximum squared overlap between any purifications of these states. Thus, if we
can show that the squared overlap between two \textit{particular}
purifications of $\omega_{A:BR}^{\text{sep}}$ and $\omega_{A:BR}$ is large,
then this implies a lower bound on the fidelity between these two states.
Consider the following particular purification of $\omega_{A:BR}^{\text{sep}}%
$:%
\[
|\omega_{ABB^{\prime}RS}^{\text{sep}}\rangle\equiv\frac{1}{\sqrt{2}}\left(
\left\vert 0\right\rangle _{A}\left\vert 0\right\rangle _{B}\left\vert
0\right\rangle _{B^{\prime}}\left\vert \psi_{\rho_{0}}\right\rangle
_{RS}+\left\vert 1\right\rangle _{A}\left\vert 1\right\rangle _{B}\left\vert
1\right\rangle _{B^{\prime}}\left\vert \psi_{\rho_{1}}\right\rangle
_{RS}\right)  .
\]
Recall that the trace distance bound in (\ref{eq:co-QSD-YES}) implies the
existence of a two-outcome projective measurement $\left\{  \Pi_{0},\Pi
_{1}\right\}  $ (known as a Helstrom measurement \cite{H69,Hol72,Hel76}) that
has the following success probability in discriminating $\rho_{0}$ from
$\rho_{1}$ if they are chosen uniformly at random:%
\begin{align}
\frac{1}{2}\text{Tr}\left\{  \Pi_{0}\rho_{0}\right\}  +\frac{1}{2}%
\text{Tr}\left\{  \Pi_{1}\rho_{1}\right\}   &  =\frac{1}{2}\left(  1+\frac
{1}{2}\left\Vert \rho_{0}-\rho_{1}\right\Vert _{1}\right)  \nonumber\\
&  \geq1-\frac{\varepsilon}{2}.\label{eq:helstrom-good}%
\end{align}
Performing the following \textquotedblleft Helstrom isometry\textquotedblright%
\[
U_{S\rightarrow SB^{\prime}}^{H}\equiv\left(  \Pi_{0}\right)  _{S}%
\otimes\left\vert 0\right\rangle _{B^{\prime}}+\left(  \Pi_{1}\right)
_{S}\otimes\left\vert 1\right\rangle _{B^{\prime}}%
\]
on the $S$ system of $\left\vert \varphi\right\rangle _{ABRS}$ produces a
particular purification of the state $\omega_{A:BR}$:%
\[
U_{S\rightarrow SB^{\prime}}^{H}\left\vert \varphi\right\rangle _{ABRS}%
=\frac{1}{\sqrt{2}}\sum_{i,j\in\left\{  0,1\right\}  }\left\vert
i\right\rangle _{A}\left\vert i\right\rangle _{B}\otimes\left\vert
j\right\rangle _{B^{\prime}}\otimes\left(  \Pi_{j}\right)  _{S}\left\vert
\psi_{\rho_{i}}\right\rangle _{RS}.
\]
The overlap between these purifications is%
\begin{align*}
&  \langle\omega_{ABB^{\prime}RS}^{\text{sep}}|U_{S\rightarrow SB^{\prime}%
}^{H}\left\vert \varphi\right\rangle _{ABRS}\\
&  =\frac{1}{2}\left(  \sum_{k\in\left\{  0,1\right\}  }\left\langle
k\right\vert _{A}\left\langle k\right\vert _{B}\left\langle k\right\vert
_{B^{\prime}}\left\langle \psi_{\rho_{k}}\right\vert _{RS}\right)  \left(
\sum_{i,j\in\left\{  0,1\right\}  }\left\vert i\right\rangle _{A}\left\vert
i\right\rangle _{B}\otimes\left\vert j\right\rangle _{B^{\prime}}%
\otimes\left(  \Pi_{j}\right)  _{S}\left\vert \psi_{\rho_{i}}\right\rangle
_{RS}\right)  \\
&  =\frac{1}{2}\sum_{i,j,k\in\left\{  0,1\right\}  }\left\langle
k|i\right\rangle _{A}\ \left\langle k|i\right\rangle _{B}\ \left\langle
k|j\right\rangle _{B^{\prime}}\ \left\langle \psi_{\rho_{k}}\right\vert
_{RS}I_{R}\otimes\left(  \Pi_{j}\right)  _{S}\left\vert \psi_{\rho_{i}%
}\right\rangle _{RS}\\
&  =\frac{1}{2}\sum_{i\in\left\{  0,1\right\}  }\left\langle \psi_{\rho_{i}%
}\right\vert _{RS}I_{R}\otimes\left(  \Pi_{i}\right)  _{S}\left\vert
\psi_{\rho_{i}}\right\rangle _{RS}\\
&  =\frac{1}{2}\text{Tr}\left\{  \Pi_{0}\rho_{0}\right\}  +\frac{1}%
{2}\text{Tr}\left\{  \Pi_{1}\rho_{1}\right\}  \\
&  \geq1-\frac{\varepsilon}{2},
\end{align*}
where the inequality follows from (\ref{eq:helstrom-good}). Squaring the
overlap gives the following lower bound on the fidelity:%
\[
F(\omega_{A:BR}^{\text{sep}},\omega_{A:BR})\geq1-\varepsilon,
\]
which imply by the Fuchs-van-de-Graaf inequalities in (\ref{eq:FvG-ineqs})
that%
\begin{equation}
\min_{\sigma_{A:BR}\in\mathcal{S}}\left\Vert \omega_{A:BR}-\sigma
_{A:BR}\right\Vert _{1}\leq2\sqrt{\varepsilon}.\label{eq:final-QSD-norm-bound}%
\end{equation}
Thus, the circuit in Figure~\ref{fig:QSD-to-QSEP} transforms a YES\ instance
of \textsf{QSD} to a YES instance of \textsf{QSEP-STATE}$_{\operatorname{1,1-LOCC}}$. We note that the
above argument is reminiscent of similar ones from quantum information theory
\cite{D05}.

We now prove that the circuit in Figure~\ref{fig:QSD-to-QSEP} transforms
NO\ instances of \textsf{QSD} into NO instances of \textsf{QSEP-STATE}$_{\operatorname{1,1-LOCC}}$. In
this case, we have the promise that the states $\rho_{0}$ and $\rho_{1}$ are
nearly indistinguishable:%
\[
\left\Vert \rho_{0}-\rho_{1}\right\Vert _{1}\leq\varepsilon.
\]
Due to the Fuchs-van-de-Graaf inequalities, we have the following lower bound
on the fidelity:%
\[
F\left(  \rho_{0},\rho_{1}\right)  \geq1-\varepsilon,
\]
and Uhlmann's theorem implies the existence of a unitary operation $U_{R}%
$\ acting on the reference system of $\left\vert \psi_{\rho_{1}}\right\rangle
_{RS}$ such that%
\[
\left\langle \psi_{\rho_{0}}\right\vert _{RS}U_{R}\otimes I_{S}\left\vert
\psi_{\rho_{1}}\right\rangle _{RS}\geq\sqrt{1-\varepsilon}.
\]
(A global phase can be fixed for $U_{R}$ such that the overlap is a real
number.)\ Thus, Bob can apply the following controlled-unitary to the state
$\left\vert \varphi\right\rangle _{ABRS}$:%
\[
C_{BR}^{U}\equiv\left\vert 0\right\rangle \left\langle 0\right\vert
_{B}\otimes I_{R}+\left\vert 1\right\rangle \left\langle 1\right\vert
_{B}\otimes U_{R},
\]
leading to%
\[
\left(  \left\vert 0\right\rangle \left\langle 0\right\vert _{B}\otimes
I_{R}+\left\vert 1\right\rangle \left\langle 1\right\vert _{B}\otimes
U_{R}\right)  \left\vert \varphi\right\rangle _{ABRS}=\frac{1}{\sqrt{2}%
}\left(  \left\vert 0\right\rangle _{A}\left\vert 0\right\rangle
_{B}\left\vert \psi_{\rho_{0}}\right\rangle _{RS}+\left\vert 1\right\rangle
_{A}\left\vert 1\right\rangle _{B}U_{R}\otimes I_{S}\left\vert \psi_{\rho_{1}%
}\right\rangle _{RS}\right)  .
\]
Then the overlap between $\left\vert \Phi^{+}\right\rangle _{AB}%
\otimes\left\vert \psi_{\rho_{0}}\right\rangle _{RS}$ and the resulting state
is large:%
\begin{align*}
&  \frac{1}{2}\left(  \left(  \left\langle 0\right\vert _{A}\left\langle
0\right\vert _{B}+\left\langle 1\right\vert _{A}\left\langle 1\right\vert
_{B}\right)  \otimes\left\langle \psi_{\rho_{0}}\right\vert _{RS}\right)
\left(  \left\vert 0\right\rangle _{A}\left\vert 0\right\rangle _{B}\left\vert
\psi_{\rho_{0}}\right\rangle _{RS}+\left\vert 1\right\rangle _{A}\left\vert
1\right\rangle _{B}U_{R}\otimes I_{S}\left\vert \psi_{\rho_{1}}\right\rangle
_{RS}\right)  \\
&  =\frac{1}{2}+\frac{1}{2}\left\langle \psi_{\rho_{0}}\right\vert _{RS}%
U_{R}\otimes I_{S}\left\vert \psi_{\rho_{1}}\right\rangle _{RS}\\
&  \geq\frac{1}{2}+\frac{1}{2}\sqrt{1-\varepsilon}\\
&  \geq\sqrt{1-\varepsilon},
\end{align*}
implying that the fidelity is larger than $1-\varepsilon$. Thus, by a local
operation, Bob can distill a state which is $2\sqrt{\varepsilon}$-close in
trace distance to the product state $\left\vert \Phi^{+}\right\rangle
_{AB}\otimes\left\vert \psi_{\rho_{0}}\right\rangle _{RS}$:%
\[
\left\Vert C_{BR}^{U}\left\vert \varphi\right\rangle \left\langle
\varphi\right\vert _{ABRS}\left(  C_{BR}^{U}\right)  ^{\dag}-\left\vert
\Phi^{+}\right\rangle \left\langle \Phi^{+}\right\vert _{AB}\otimes\left\vert
\psi_{\rho_{0}}\right\rangle \left\langle \psi_{\rho_{0}}\right\vert
_{RS}\right\Vert _{1}\leq2\sqrt{\varepsilon}.
\]
(We remark that the above argument is similar to a \textquotedblleft
decoupling\textquotedblright\ argument well known in quantum information
theory \cite{D05,ADHW06FQSW}.)

Now, we would like to argue that the one-way LOCC\ distance between
$\omega_{A:BR}$ and the separable state $\sigma_{A:BR}^{\ast}\in S$ closest to
$\omega_{A:BR}$ is larger than an appropriate constant, so that we can claim
that the circuit in Figure~\ref{fig:QSD-to-QSEP}\ maps NO instances of
\textsf{QSD}\ to NO instances of \textsf{QSEP-STATE}$_{\operatorname{1,1-LOCC}}$. In order to do so, Bob
first performs the local unitary $C_{BR}^{U}$. This transforms the state
$\left(  C_{BR}^{U}\right)  ^{\dag}\left(  \Phi_{AB}^{+}\otimes\left(
\psi_{\rho_{0}}\right)  _{R}\right)  C_{BR}^{U}$ to $\Phi_{AB}^{+}%
\otimes\left(  \psi_{\rho_{0}}\right)  _{R}$ and the separable state
$\sigma_{A:BR}^{\ast}$ to some other separable state $\left(  \sigma
_{A:BR}^{\ast}\right)  ^{\prime}$. Alice and Bob then perform a
Bell experiment, guessing the state to be $\left\vert \Phi^{+}\right\rangle
_{AB}$ if there is a violation of a Bell inequality and guessing a separable
state otherwise \cite{bell1964}. Equivalently, Alice and Bob could proceed as
in the CHSH game (a reformulation of a Bell experiment as a nonlocal game
\cite{CHTW04}). In such a protocol, Alice flips a coin $x$\ and chooses one of
two binary-outcome measurements to perform on her qubit. She sends both $x$
and the measurement outcome $a$ to Bob. Bob then flips a coin with outcome $y$
and performs one of two binary-outcome measurements on his qubit, naming the
measurement result $b$. Bob declares the state to be the Bell state in the
case that $x\wedge y=a\oplus b$ (when they \textquotedblleft win the CHSH
game\textquotedblright) and otherwise declares that it is not the Bell state.
It is well known that the winning probability of the CHSH\ game with a Bell
state is equal to $\cos^{2}\left(  \pi/8\right)  \approx0.85$, while the
maximum probability with which they can win this game with a separable state
is equal to $0.75$ \cite{CHTW04}. This gives the following lower bound on
the one-way LOCC\ distance between $\left(  C_{BR}^{U}\right)  ^{\dag}\left(
\Phi_{AB}^{+}\otimes\left(  \psi_{\rho_{0}}\right)  _{R}\right)  C_{BR}^{U}$
and $\sigma_{A:BR}^{\ast}$:%
\begin{align}
\left\Vert \left(  C_{BR}^{U}\right)  ^{\dag}\left(  \Phi_{AB}^{+}%
\otimes\left(  \psi_{\rho_{0}}\right)  _{R}\right)  C_{BR}^{U}-\sigma
_{A:BR}^{\ast}\right\Vert _{1-\text{LOCC}} &  =\left\Vert \Phi_{AB}^{+}%
\otimes\left(  \psi_{\rho_{0}}\right)  _{R}-\left(  \sigma_{A:BR}^{\ast
}\right)  ^{\prime}\right\Vert _{1-\text{LOCC}}\nonumber\\
&  \geq\left\Vert \left(  \cos^{2}\left(  \pi/8\right)  ,\sin^{2}\left(
\pi/8\right)  \right)  -\left(  0.75,0.25\right)  \right\Vert _{1}\nonumber\\
&  \geq0.2.\label{eq:cool-CHSH-bound}%
\end{align}
Thus, by combining with the distillation argument above, we have the following
lower bound on the one-way LOCC\ distance between $\omega_{A:BR}$ and
$\sigma_{A:BR}^{\ast}$:%
\begin{align}
\left\Vert \omega_{A:BR}-\sigma_{A:BR}^{\ast}\right\Vert _{1-\text{LOCC}} &
\geq\left\Vert \left(  C_{BR}^{U}\right)  ^{\dag}\left(  \Phi_{AB}^{+}%
\otimes\left(  \psi_{\rho_{0}}\right)  _{R}\right)  C_{BR}^{U}-\sigma
_{A:BR}^{\ast}\right\Vert _{1-\text{LOCC}}\nonumber\\
&  \ \ \ \ \ -\left\Vert \left(  C_{BR}^{U}\right)  ^{\dag}\left(  \Phi
_{AB}^{+}\otimes\left(  \psi_{\rho_{0}}\right)  _{R}\right)  C_{BR}^{U}%
-\omega_{A:BR}\right\Vert _{1-\text{LOCC}}\nonumber\\
&  \geq\left\Vert \Phi_{AB}^{+}\otimes\left(  \psi_{\rho_{0}}\right)
_{R}-\left(  \sigma_{A:BR}^{\ast}\right)  ^{\prime}\right\Vert _{1-\text{LOCC}%
}\nonumber\\
&  \ \ \ \ \ -\left\Vert \left(  C_{BR}^{U}\right)  ^{\dag}\left(  \Phi
_{AB}^{+}\otimes\left(  \psi_{\rho_{0}}\right)  _{R}\right)  C_{BR}^{U}%
-\omega_{A:BR}\right\Vert _{1}\nonumber\\
&  \geq0.2-2\sqrt{\varepsilon},\label{eq:final-QSD-NO-bound}%
\end{align}
where the second inequality follows from (\ref{eq:difference-trace-LOCC}) and
the fact that $\left(  C_{BR}^{U}\right)  $ is a local unitary, and the third
from (\ref{eq:cool-CHSH-bound}) and the argument at the end of the previous
paragraph. Thus, as long as $\varepsilon$ is small enough (so that
$0.2-4\sqrt{\varepsilon}>0$), there is a gap between
(\ref{eq:final-QSD-norm-bound}) and (\ref{eq:final-QSD-NO-bound}). In fact,
Watrous showed that it is possible to make $\varepsilon$ exponentially small
with only polynomial overhead for any instance of \textsf{QSD}\ \cite{W02} by
exploiting a \textquotedblleft quantized\textquotedblright\ version of the
polarization lemma in \cite{Sahai:1997}. Thus, any protocol for deciding
\textsf{QSEP-STATE}$_{\operatorname{1,1-LOCC}}$\ could also decide \textsf{QSD}, implying that
\textsf{QSEP-STATE}$_{\operatorname{1,1-LOCC}}$\ is \textsf{QSZK}-hard.
\end{proof}

Ideally, we would like to show that \textsf{QSEP-STATE}$_{\operatorname{1,1-LOCC}}$\ is a complete
promise problem for \textsf{QIP(2)}, but it is not clear to us how to do so.
The obvious way to attempt this would be to reduce \textsf{CLOSE-IMAGE}\ to
\textsf{QSEP-STATE}$_{\operatorname{1,1-LOCC}}$, but the problem is that \textsf{CLOSE-IMAGE}\ requires
a general channel, whereas our protocol for \textsf{QSEP-STATE}$_{\operatorname{1,1-LOCC}}$\ has a very
specific channel (one that applies a random permutation to the $B$ systems and
discards the last $k-1$ of them). Alternatively, we could attempt to find a
\textsf{QSZK}\ proof system for \textsf{QSEP-STATE}$_{\operatorname{1,1-LOCC}}$, but the protocol that
we have given to show that \textsf{QSEP-CIRCUIT~}$\in~$\textsf{QIP(2)} does
not satisfy the zero-knowledge property because, in the case of a
YES\ instance, the verifier ends up with a state close to a $k$-extension of
$\rho_{AB}$, which he could not have generated himself using a polynomial-time
quantum circuit.

\section{\textsf{NP}-hardness of the quantum separability problem}

\label{sec:qsep-np-hard}We now prove \textsf{NP}-hardness of
\textsf{QSEP-STATE}$_{\operatorname{1,1-LOCC}}$, with respect to Cook reductions,
by finding a reduction to it from the \textsf{NP}-hard
matrix version of the quantum separability problem. The essence of the
reduction is Knill's efficient encoding of a density matrix description of a
state $\rho_{AB}$ as a description of a quantum circuit to generate it
\cite{Kn95}. We begin by recalling the matrix version of the quantum
separability problem:

\begin{problem}
[WMEM$_{\varepsilon}\left(  M,N\right)  $]Given a density matrix $\rho_{AB}%
\in\mathcal{D}\left(  \mathcal{H}_{M}\otimes\mathcal{H}_{N}\right)  $ with rational entries subject to the
promise that either (i) $\rho_{AB}\in\mathcal{S}$ or (ii) $\min_{\sigma
_{AB}\in\mathcal{S}}\left\Vert \rho_{AB}-\sigma_{AB}\right\Vert _{2}%
\geq\varepsilon$, with $\varepsilon$ no smaller than an inverse polynomial in
$MN$, decide which is the case.
\end{problem}

Gharibian showed that the above promise problem is \textsf{NP}-hard
\cite{G10}, so our task is just to find a Cook reduction from
WMEM$_{\varepsilon}\left(  M,N\right)  $ to \textsf{QSEP-STATE}$_{\operatorname{1,1-LOCC}}$. First,
consider that we can diagonalize the matrix $\rho_{AB}$ in time polynomial in
$MN\log( MN /\varepsilon_{1})$, where $\varepsilon_{1}$ is an error parameter
characterizing the precision of the diagonalization in the trace distance.
We then compute a purification
$\left\vert \phi_{\rho}\right\rangle _{RAB}$\ of $\rho_{AB}$ to a reference
system with dimension no larger than $MN$. Knill's algorithm gives a quantum
circuit running on $O\left(  \log\left(  MN\right)  \right)  $ qubits that
generates the state $\left\vert \phi_{\rho}\right\rangle ^{RAB}$ \cite{Kn95},
and this algorithm runs in time polynomial in $MN$. Knill's algorithm outputs
controlled single-qubit unitary gate descriptions with arbitrary precision, so
we need to invoke the Solovay-Kitaev algorithm \cite{DN06}\ to approximate
each gate in Knill's circuit with unitaries chosen from a finite gate set, up
to precision $\varepsilon_{2}/l$ where $l$ is the number of gates in Knill's
circuit. The Solovay-Kitaev algorithm runs in time
polylogarithmic in $l/\varepsilon_{2}$ and produces a gate sequence with
length polylogarithmic in $l/\varepsilon_{2}$. This whole procedure leads to a mixed state quantum circuit generating a state $\rho_{AB}'$ such that $\| \rho_{AB} - \rho_{AB}' \|_1 \leq \varepsilon_1 + \varepsilon_2$. The state $\rho_{AB}'$ will be used as the input to \textsf{QSEP-STATE}$_{\operatorname{1,1-LOCC}}$$(\delta_c,\delta_s)$.

Setting $\delta_c = \varepsilon_1 + \varepsilon_2$ implies that any instance of $\textsf{WMEM}_\varepsilon$ for which $\rho_{AB} \in \mathcal{S}$, meaning case \emph{(i)} of the promise, gets mapped to a YES instance of \textsf{QSEP-STATE}$_{\operatorname{1,1-LOCC}}$$(\delta_c,\delta_s)$. For case \emph{(ii)}, we know from Matthews
\textit{et al}.~\cite{MWW09}\ that%
\begin{equation}
\min_{\sigma_{AB}\in\mathcal{S}}\left\Vert \rho_{AB}-\sigma_{AB}\right\Vert
_{1-\text{LOCC}}\geq\frac{1}{\sqrt{153}}\min_{\sigma_{AB}\in\mathcal{S}%
}\left\Vert \rho_{AB}-\sigma_{AB}\right\Vert _{2}\geq\frac{\varepsilon}%
{\sqrt{153}}. \label{eq:matthews-bound}
\end{equation}
This in turn implies that
\[
\min_{\sigma_{AB} \in \mathcal{S}} \| \rho_{AB}' - \sigma_{AB} \|_{1-\text{LOCC}}
\geq \frac{\varepsilon}{\sqrt{153}} - \varepsilon_1 - \varepsilon_2
\]
so if we choose $\delta_s = \varepsilon/\sqrt{153} - \varepsilon_1 - \varepsilon_2$ then case \emph{(ii)} gets mapped to a NO instance of \textsf{QSEP-STATE}$_{\operatorname{1,1-LOCC}}$$(\delta_c,\delta_s)$. Moreover, because $\varepsilon_1 + \varepsilon_2$ can be made to shrink exponentially with the circuit size, the gap $\delta_s - \delta_c$ remains inverse polynomial in the circuit size. In particular, the instance of \textsf{QSEP-STATE}$_{\operatorname{1,1-LOCC}}$$(\delta_c,\delta_s)$ will be in $\textsf{QIP(2)}$ for sufficiently small $\varepsilon$, as determined by the promise in Theorem~\ref{thm:QSEP-in-QIP(2)}.

\section{\textsf{QIP}-completeness of the channel quantum separability
problem}

\label{sec:channel-qsep-qip-complete}There is a straightforward variation of
\textsf{QSEP-STATE}$_{\operatorname{1,1-LOCC}}$ which is a complete promise problem for \textsf{QIP(3)}
(and therefore complete for \textsf{QIP} \cite{KW00}).\ In this variation, the
input is a description of a circuit that implements a quantum channel with
input system $S$ and two output systems $A$ and $B$. (The channel is
implemented by a unitary circuit with qubits in an environment system $R$ that
are traced over.) Figure~\ref{fig:circuit-for-channel}\ depicts a circuit that
implements such a channel. The task is to decide whether there is an input to
the channel such that the output state on systems $A$ and $B$ is separable.%
\begin{figure}
[ptb]
\begin{center}
\includegraphics[
natheight=1.166600in,
natwidth=1.793600in,
width=1.8308in
]%
{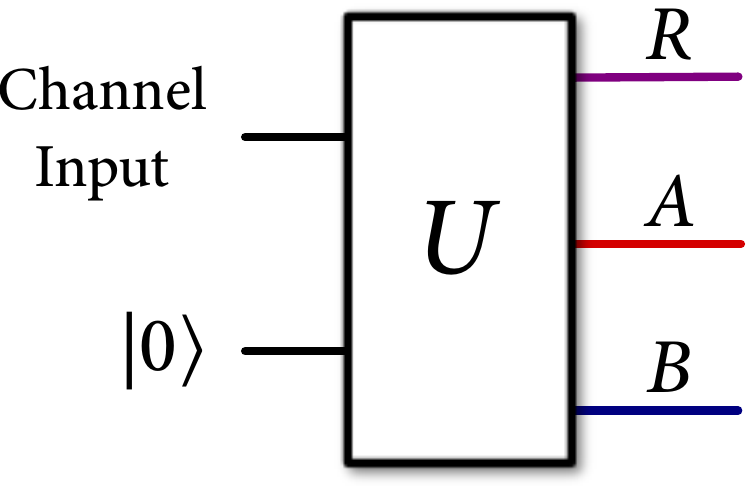}%
\caption{A quantum circuit to implement a channel. The circuit has input
qubits and ancillas in the state $\left\vert 0\right\rangle $. The circuit
outputs qubits in the environment system $R$ (which are traced over) and
qubits in systems $A$ and $B$.}%
\label{fig:circuit-for-channel}%
\end{center}
\end{figure}

\begin{problem}
[\textsf{QSEP-CHANNEL}$_{\operatorname{1,1-LOCC}}$$\left(  \delta_{c},\delta_{s}\right)  $]%
\label{prob:QSEP-CHANNEL}Given is a mixed-state quantum circuit to generate
the channel $\mathcal{N}_{S\rightarrow AB}$, having an $n$-qubit input and an
$m$-qubit output, along with a labeling of the qubits in the environment
system $R$ and the output qubits for $A$ and $B$. Decide whether

\begin{enumerate}
\item Yes:\ There is an input to the channel $\rho_{S}$ such that the channel
output $\mathcal{N}_{S\rightarrow AB}\left(  \rho_{S}\right)  $ is $\delta
_{c}$-close in trace distance to a separable state $\sigma_{AB}\in\mathcal{S}%
$:%
\begin{equation}
\min_{\rho_{S},\ \sigma_{AB}\in\mathcal{S}}\left\Vert \mathcal{N}%
_{S\rightarrow AB}\left(  \rho_{S}\right)  -\sigma_{AB}\right\Vert _{1}%
\leq\delta_{c}.\label{eq:YES-instance-QSEP-CHANNEL}%
\end{equation}

\item No:\ For all channel inputs $\rho_{S}$, the channel output
$\mathcal{N}_{S\rightarrow AB}\left(  \rho_{S}\right)  $ is at least $\delta_{s}$-far
in 1-LOCC$\ $distance to a separable state:%
\[
\min_{\rho_{S},\ \sigma_{AB}\in\mathcal{S}}\left\Vert \mathcal{N}%
_{S\rightarrow AB}\left(  \rho_{S}\right)  -\sigma_{AB}\right\Vert
_{1-\text{LOCC}}\geq\delta_{s}.
\]

\end{enumerate}
\end{problem}

\begin{theorem}
\textsf{QSEP-CHANNEL}$_{\operatorname{1,1-LOCC}}$$\left(  \delta_{c},\delta_{s}\right)  $ is
\textsf{QIP}-complete if there are polynomial-time computable functions
$\delta_{c},\delta_{s}:\mathbb{N}\rightarrow\left[  0,1\right]  $, such that
the difference $\delta_s^2 / 8  -2\sqrt{\delta_{c}}$ is larger
than an inverse polynomial in the circuit size.
\end{theorem}

\begin{proof}
The proof of this theorem is almost identical to the proofs of
Theorems~\ref{thm:QSEP-in-QIP(2)} and \ref{thm:QSEP-QSZK-hard}.

We first show that there is a three-message quantum interactive proof system
for \textsf{QSEP-CHANNEL}$_{\operatorname{1,1-LOCC}}$. This is just the obvious modification of the
circuit in Figure~\ref{fig:qsep-qip2} so that it becomes a three-message proof
system as in Figure~\ref{fig:qip-3-and-2}(a). In particular, the prover first
prepares a state and sends it to the verifier. The verifier inputs this state
to the circuit that implements the channel $\mathcal{N}_{S\rightarrow AB}$,
and the rest of the proof system proceeds as in Figure~\ref{fig:qsep-qip2}. In
the case of a positive instance, the prover can compute the states $\rho_{S}$
and $\sigma_{AB}$ in (\ref{eq:YES-instance-QSEP-CHANNEL}) from the description
of the channel $\mathcal{N}_{S\rightarrow AB}$. He generates $\rho_{S}$ with
his first unitary operation and then proceeds by choosing his second unitary
operation as if the state $\mathcal{N}_{S\rightarrow AB}\left(  \rho
_{S}\right)  $ were $\sigma_{AB}$. Following the same analysis as in the proof
of Theorem~\ref{thm:QSEP-in-QIP(2)}, the maximum probability with which the
verifier accepts in this case is no smaller than $1-2\sqrt{\delta_{c}}$. In
the case of a negative instance, by Lemma~\ref{cor:contra-approx-k-ext}, for
every state $\mathcal{N}_{S\rightarrow AB}\left(  \rho_{S}\right)  $, there is
some $k$ polynomial in the circuit size such that the maximum probability with
which the prover can make the verifier accept is no larger than $1-\delta_s^2 / 8  $.
An upper bound on the maximum acceptance probability is%
\[
\max_{\omega_{S},\ \sigma_{AB}\in\mathcal{E}_{k}}F\left(  \mathcal{N}%
_{S\rightarrow AB}\left(  \omega_{S}\right)  ,\sigma_{AB}\right)  ,
\]
a formula which follows from (\ref{eq:max-output-fidelity}) and our previous
analysis in the proof of Theorem~\ref{thm:QSEP-in-QIP(2)}. This leaves a gap
of $\delta_s^2 / 8  -2\sqrt{\delta_{c}}$ between completeness and
soundness error (promised to be larger than an inverse polynomial) and it is
known that this gap can be amplified \cite{KW00}. Thus, \textsf{QSEP-CHANNEL}$_{\operatorname{1,1-LOCC}}$%
$\left(  \delta_{c},\delta_{s}\right)  $\textsf{\ }$\in\ $\textsf{QIP}.

To show that \textsf{QSEP-CHANNEL\ }is$\ $\textsf{QIP}-hard,\ it suffices to
exhibit a reduction from the \textsf{QIP}-complete promise problem
\textsf{QCD}\ (Problem~\ref{prob:QCD}) to \textsf{QSEP-CHANNEL}$_{\operatorname{1,1-LOCC}}$. This
reduction is essentially the same as that in the proof of
Theorem~\ref{thm:QSEP-QSZK-hard}, except that the circuit in
Figure~\ref{fig:QSD-to-QSEP}\ is modified so that the unitaries being
controlled are the unitaries that generate the channels (rather than the ones
that generate the states). In the case of a positive instance of \textsf{QCD},
there exists an input to the channels such that their outputs are nearly
distinguishable, so that the output of the modified circuit is nearly
separable. Also, in the case of a negative instance, the outputs of the
channels for all inputs are nearly indistinguishable, so that it is possible
to distill a Bell state from the output state of the modified circuit. The
CHSH game argument then applies as well. Thus, \textsf{QSEP-CHANNEL\ }%
is$\ $\textsf{QIP}-hard.
\end{proof}

\section{A two-message quantum interactive proof system decides the multipartite quantum
separability problem}

\label{sec:qmultisep-in-qip2}
After their seminal work on $k$-extendibility as a test of separability
for bipartite states \cite{DPS02,DPS04},
Doherty {\it et al}.~developed a notion of $k$-extendibility for multipartite quantum states \cite{DPS05}.
Recently, Brand\~{a}o and Christandl exploited this notion to construct
a quasi-polynomial time algorithm that decides a variant of the multipartite
quantum separability problem \cite{BC12}.
Brand{\~{a}}o and Harrow then improved the runtime of the algorithm by proving a stronger
quantum de~Finetti theorem that is applicable to the multipartite case \cite{BH12}.

In this section, we extend our results 
from Section~\ref{sec:qsep-in-qip2} to the multipartite case by using the results in \cite{BH12}. In particular,
we formulate a variant of the multipartite quantum separability problem
that we name \textsf{MULTI-QSEP-STATE}$_{\operatorname{1,1-LOCC}}$.

We begin with a few definitions before proceeding to the main theorem of this section.
A multipartite quantum state with $l$ parties $A_1, \ldots, A_l$ is fully separable if it can be written in the following form:
\begin{equation}
\sigma_{A_1 : \, \cdots\,  : A_l}=\sum_{x\in\mathcal{X}}p_{X}\left(  x\right)  \ \vert \psi^1
_{x}\rangle \langle \psi^1_{x}\vert _{A_1}\otimes
\cdots \otimes
\vert
\psi^l_{x}\rangle \langle \psi^l_{x}\vert _{A_l} \label{eq:fully-sep}
\end{equation}
Let $\mathcal{S}$ denote the set of $l$-partite fully separable states. 
(We omit the dependence on the number of parties, $l$, which should be clear from context.)

The notion of $k$-extendibility extends in a natural way to multipartite systems.
For notational simplicity, 
we refer to the total system which we are extending as $C$, and the $l$ subsystems of $C$ as 
$A_1, A_2, \dots, A_l$. For example, we abbreviate $\sigma_{A_1 :\, \cdots\, : A_l}$ simply as
$\sigma_{C}$.
A multipartite state $\rho_C = \rho_{A_1:\,\cdots\, : A_l}\in\mathcal{D}\left(  \mathcal{H}_{A}%
\otimes \cdots \otimes \mathcal{H}_{A_l}\right)  $
is $k$-extendible \cite{DPS05}\ if
there exists a state $\omega_{C C_{2}\cdots C_{k}}\in\mathcal{D}\left(
\mathcal{H}_{C}\otimes\mathcal{H}_{C_{2}}\otimes\cdots\otimes\mathcal{H}%
_{C_{k}}\right)  $ such that

\begin{enumerate}
\item Each Hilbert space $\mathcal{H}_{C_{i,j}}$ is isomorphic to $\mathcal{H}_{A_j}$
for all $i\in\left\{  1,\ldots,k\right\}  $ and $j \in \{1, \ldots, l\}$. (We are using
the notation $C_{i,j}$ to refer to the $j^{\text{th}}$ subsystem of $C_i$.)

\item For all parties $j \in \{1, \ldots, l\}$, the state
$\omega_{C C_{2}\cdots C_{k}}$ is invariant under permutations
of the systems $C_{1,j}$ through $C_{k,j}$.
Note that there are $l\cdot k!$ such permutations.

\item The state $\omega_{C C_{2}\cdots C_{k}}$ is an extension of $\rho_C$:%
\[
\rho_C=\text{Tr}_{C_{2}\cdots C_{k}}\left\{  \omega_{C C_{2}\cdots C_{k}} \right\}  .
\]

\end{enumerate}

Let $\mathcal{E}_{k}$ denote the set of $k$-extendible states for $l$ parties (again
suppressing the dependence of $\mathcal{E}_{k}$ on $l$ as it should be clear from context).
A fully separable state
$\sigma_{A_1 : \, \cdots\,  : A_l}$ of the form in (\ref{eq:fully-sep})
has a $k$-extension of the following form for all $k$:
\begin{equation}
\sum_{x\in\mathcal{X}}p_{X}\left(  x\right)  \ \left(\vert \psi^1
_{x}\rangle \langle \psi^1_{x}\vert _{A_1}\right)^{\otimes k}\otimes
\cdots \otimes
\left(\vert
\psi^l_{x}\rangle \langle \psi^l_{x}\vert _{A_l}\right)^{\otimes k},
\end{equation}
which can be purified as
\begin{equation}
\ket{\phi_{\sigma,k}}_{R^{\prime} C C_2 \cdots C_k} \equiv \sum_{x\in\mathcal{X}}\sqrt{p_{X}\left(  x\right)} \vert x\rangle_{R^{\prime}}  \ \left(\vert \psi^1
_{x}\rangle_{A_1} \right)^{\otimes k}\otimes
\cdots \otimes
\left(\vert
\psi^l_{x}\rangle _{A_l}\right)^{\otimes k}%
\end{equation}

\begin{problem}
[\textsf{MULTI-QSEP-STATE}$_{\operatorname{1,1-LOCC}}$$\left(  \delta_{c},\delta_{s}\right)  $]%
\label{prob:QMULTISEP-CIRCUIT}Given is a mixed-state quantum circuit to generate
the $n$-qubit state $\rho_{C}$, along with a labeling of the qubits in the
reference system $R$ and the output qubits for each system $A_1, \ldots, A_l \in C$. 
Decide whether

\begin{enumerate}
\item Yes:\ There is a fully separable state $\sigma_{C}\in\mathcal{S}$\ that is $\delta_c$-close to $\rho_C$ in trace distance:%
\[
\min_{\sigma_{C}\in\mathcal{S}}\left\Vert \rho_{C}-\sigma_{C}\right\Vert
_{1}\leq\delta_{c}.
\]

\item No:\ All fully separable states are at least $\delta_{s}%
$-far from $\rho_C$ in 1-LOCC$\ $distance:%
\[
\min_{\sigma_{C}\in\mathcal{S}}\left\Vert \rho_{C}-\sigma_{C}\right\Vert
_{1-\text{LOCC}}\geq\delta_{s}.
\]

\end{enumerate}
\end{problem}
The promise on negative instances is in terms of the multipartite
1-LOCC distance defined in \cite{BH12}:
\begin{equation}
\left\Vert \rho_{A_1 : \, \cdots \, : A_l}-\sigma_{A_1 : \, \cdots \, : A_l} \right\Vert _{1-\text{LOCC}}
\equiv
\max_{\Lambda_2, \cdots, \Lambda_l} \left\Vert \left(  I_{A_1}\otimes\Lambda_2
\otimes \cdots \otimes  \Lambda_l \right) \left(  \rho_{A_1 : \, \cdots \, : A_l}
-\sigma_{A_1 : \, \cdots \, : A_l}\right)  \right\Vert
_{1}, \label{eq:multi-1-LOCC}
\end{equation}
where $\Lambda_2, \cdots, \Lambda_l$ are quantum-to-classical channels.

\begin{theorem}
\label{thm:QMULTISEP-in-QIP(2)}\textsf{MULTI-QSEP-STATE}$_{\operatorname{1,1-LOCC}}$$\left(  \delta_{c},\delta
_{s}\right)  $\textsf{\ }$\in\ $\textsf{QIP(2)} if there are polynomial-time
computable functions $\delta_{c},\delta_{s}:\mathbb{N}\rightarrow\left[
0,1\right]  $, such that the difference $\delta_{s}^2 / 8  -2\sqrt{\delta_{c}}$
is larger than an inverse polynomial in the circuit size.\end{theorem}

\begin{proof}
The proof system for \textsf{MULTI-QSEP-STATE}$_{\operatorname{1,1-LOCC}}$ is similar to that of \textsf{QSEP-STATE}$_{\operatorname{1,1-LOCC}}$,
in the sense that it amounts to a quantum computational test for multipartite $k$-extendibility.
We exploit a generalized version of Lemma~\ref{cor:contra-approx-k-ext} that applies to
multipartite states. This lemma follows from Theorem~2 of
Brand{\~{a}}o and Harrow \cite{BH12}, and we provide a proof for it in the appendix.

\begin{lemma}
\label{cor:contra-approx-multi-k-ext}Let $\rho_{C}$ be $\varepsilon$-away in
one-way LOCC distance from the set of fully separable states, for some
$\varepsilon>0$:%
\[
\min_{\sigma_{C} \in\mathcal{S}}\left\Vert \rho_{C}-\sigma_{C}\right\Vert
_{\text{1-LOCC}}\geq\varepsilon.
\]
Then the state $\rho_{C}$ is $\delta$-away in trace distance from the set of
$k$-extendible states:%
\[
\min_{\sigma_{C}\in\mathcal{E}_{k}}\left\Vert \rho_{C}-\sigma_{C}%
\right\Vert _{1}\geq\delta,
\]
for $\delta < \varepsilon$ and where%
\[
k=\left\lceil l + \frac{4 l^2\log\left\vert C\right\vert }{(\varepsilon
-\delta)^2  }\right\rceil .
\]
\end{lemma}

Figure~\ref{fig:qmultisep-qip2}\ depicts a two-message quantum interactive proof
system for \textsf{MULTI-QSEP-STATE}$_{\operatorname{1,1-LOCC}}$. The protocol begins with the verifier
preparing the state $\left\vert \psi_{\rho}\right\rangle _{RC}$, a particular
purification of $\rho_{C}$, by running the quantum circuit $U_{\rho}
$ as given in the problem instance. The verifier transmits the reference system 
to the prover, who then acts on $R$ and some
ancillary qubits with a unitary $P_{1}$\ that has output systems $R^{\prime}$,
$C_{2}$, \ldots, $C_{k}$. The prover transmits systems $C_{2}$, \ldots,
$C_{k}$ to the verifier. The verifier then performs phase estimation over the
symmetric group \cite{K95,BBDEJM97} (also known as the \textquotedblleft
permutation test\textquotedblright\ \cite{KNY08}) on the registers $C$,
$C_{2}$, \ldots, $C_{k}$, using the qubits in system $D$ as the control.
This control register requires $\lceil \log(l\cdot k!) \rceil$ qubits because
 the permutations included in the test are those from the definition of multipartite
$k$-extendibility. The
verifier performs a computational basis measurement on all of the qubits in
the control register $D$ and accepts if and only if the measurement outcome
is all zeros. This protocol operates nearly exactly as the bipartite case does, except that
the verifier asks for $k$-extensions of all the systems instead of just one. The analysis
of the protocol proceeds almost identically to the analysis in Section \ref{sec:qsep-in-qip2},
however we state it here explicitly for clarity.
\begin{figure}
[ptb]
\begin{center}
\includegraphics[
natheight=2.040100in,
natwidth=4.427000in,
width=4.478in
]%
{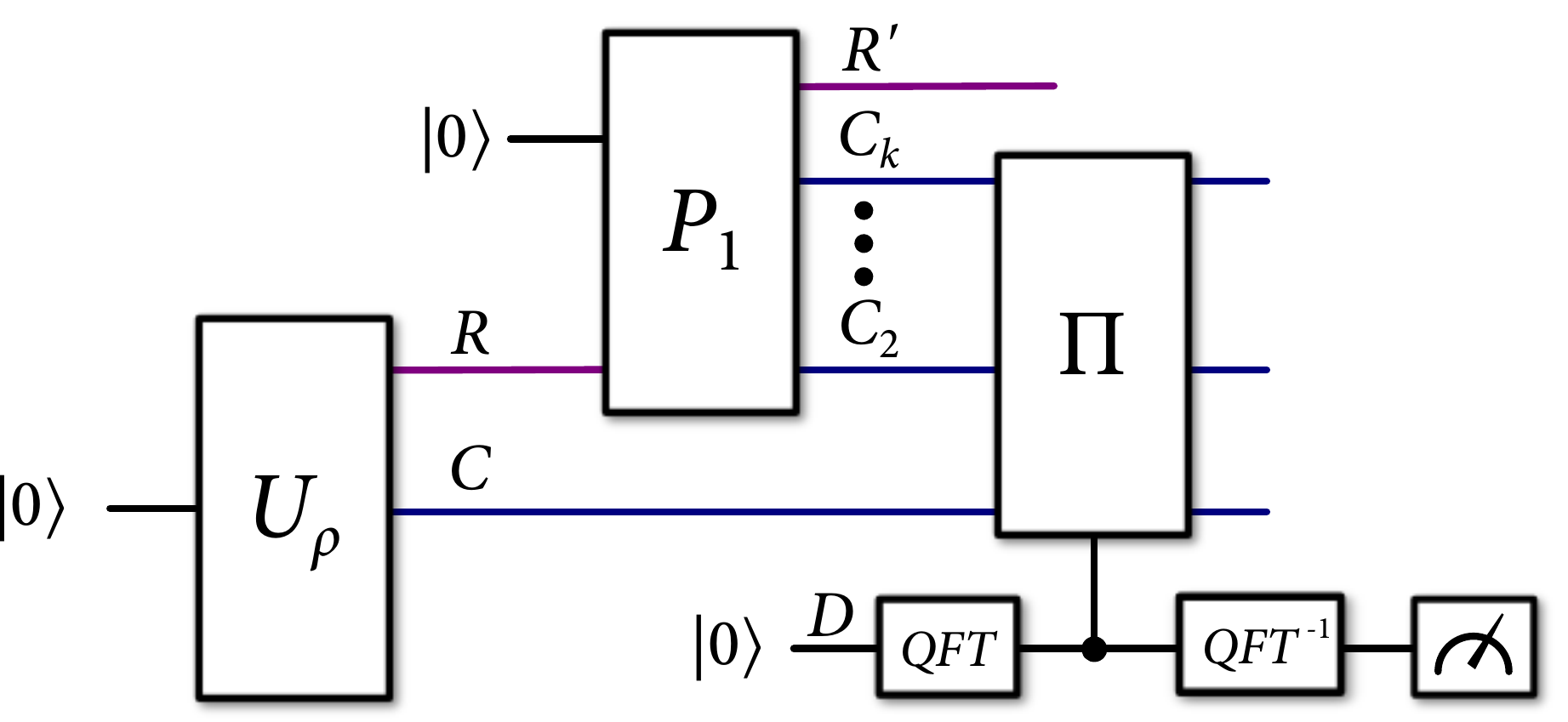}%
\caption{A two-message quantum interactive proof system for
\textsf{MULTI-QSEP-STATE}$_{\operatorname{1,1-LOCC}}$. As in the bipartite case, the
protocol begins with the verifier executing the circuit
$U_{\rho}$ that generates the state $\rho_{C}$. He sends the reference system
to the prover. In the case that $\rho_{C}$ is fully separable, the prover should be
able to act with a unitary on the reference system and some ancillas in order
to generate a multipartite $k$-extension of $\rho_C$ to 
the systems $C_2$ through $C_k$. The prover sends
all of the extension systems back to the verifier, who then performs phase
estimation over the symmetric group (a quantum Fourier transform followed by a
controlled permutation) in order to test if the state sent by the prover is
a multipartite $k$-extension.} 
\label{fig:qmultisep-qip2}%
\end{center}
\end{figure}

 In a YES instance, 
there is some fully separable state $\sigma_C \in \mathcal{S}$ that is $\delta_c$-close in 
trace distance to $\rho_C$. By Uhlmann's theorem in (\ref{eq:uhlmann-thm}) and the
Fuchs-van-de-Graff inequalities in (\ref{eq:FvG-ineqs}), there exists a purification $\ket{\psi_\sigma}$
of $\sigma_C$ such that 
\begin{equation}\label{eq:multiYES-instance-err}
\left\Vert \ket{\psi_\rho}\bra{\psi_\rho}_{RC} - \ket{\psi_\sigma}\bra{\psi_\sigma}_{RC}\right\Vert
\le 2\sqrt{\delta_c}.
\end{equation}
As such, the prover can operate by providing a $k$-extension for the separable state $\sigma_C$
instead, giving a lower bound on the probability that the verifier accepts. Letting $U$ be the unitary
that includes the prover's operation $P_1$ and the verifier's permutation test, we have that
\begin{equation*}
\operatorname{Tr}\left\{ \ket{0}\bra{0}_D U ( \ket{\psi_\rho}\bra{\psi_\rho}_{RC} ) U^\dagger \right\}
\ge 1 - 2\sqrt{\delta_c},
\end{equation*}
where the inequality follows exactly the same line of reasoning as the steps in (\ref{eq:QSEP-YES-analysis})
by exploiting (\ref{eq:multiYES-instance-err}) instead.

The analysis of a NO instance also proceeds similarly to that of the bipartite case, but the 
channel generated by the inverse of the verifier's circuit conditional on accepting is now given by
\[
\mathcal{M}_{CC_{2}\cdots C_{k}\rightarrow CD}\left(  \sigma_{CC_{2}\cdots
C_{k}}\right)  \equiv\text{Tr}_{C_{2}\cdots C_{k}}\left\{  \left(  U_{\Pi
}\right)  _{CC_{2}\cdots C_{k}D}\left(  \sigma_{CC_{2}\cdots C_{k}}%
\otimes\left\vert \text{perm}\right\rangle \left\langle \text{perm}\right\vert
_{D}\right)  \left(  U_{\Pi}^{\dag}\right)  _{CC_{2}\cdots C_{k}D}\right\}  ,
\]
where $(U_\Pi)_{CC_2\cdots C_kD}$ is a controlled-permutation operation 
defined similarly to that in (\ref{eq:controlled-perm}), and the permutations involved
are those from the definition of multipartite $k$-separability.
Also, $\left\vert \text{perm}\right\rangle$ is defined similarly to (\ref{eq:perm-state}), though
it is a uniform superposition of all possible $l\cdot k!$ permutations required from
the definition of multipartite $k$-extendibility.
Tracing out the control register $D$ 
gives the following upper bound on the acceptance probability:
\begin{equation} \label{eq:multifidelity-QIP(2)}
\max_{\sigma_{CC_2\cdots C_k}} F(\rho_{C}, \mathcal{M}_{CC_2\dots C_k\rightarrow C}(\sigma_{CC_2\cdots C_k}),
\end{equation}
where $\mathcal{M}_{CC_2\dots C_k\rightarrow C}$
is a channel that applies a permutation selected at random from the multipartite
$k$-extendibility permutations and  then
discards the systems $C_2$, \ldots, $C_k$. As in the bipartite case, since the channel 
$\mathcal{M}_{CC_2\cdots C_k\rightarrow C}$ symmetrizes the state of the 
systems $CC_2\cdots C_k$, the maximum in (\ref{eq:multifidelity-QIP(2)}) is 
achieved by a state $\sigma_{CC_2\cdots C_k}$ which is permutation symmetric 
with respect to the multipartite
$k$-extendibility permutations. As such, we can rewrite (\ref{eq:multifidelity-QIP(2)}) as the 
maximum $k$-extendible fidelity of $\rho_{C}$:
\begin{equation}
\max_{\sigma_{CC_2\cdots C_k}} F(\rho_{C}, \mathcal{M}_{CC_2\dots C_k\rightarrow C}(\sigma_{CC_2\cdots C_k}) = \max_{\sigma_{C} \in \mathcal{E}_k} F(\rho_{C}, \sigma_{C}).
\end{equation}

By Lemma~\ref{cor:contra-approx-multi-k-ext}, if we take $k$ to be larger than 
\[
\left\lceil l + \frac{4 l^2\log\left\vert C\right\vert }{(\delta_s-\delta'_s)^2  }\right\rceil ,
\]
then 
\[ \min_{\sigma_{C} \in \mathcal{E}_k} \left\Vert \rho_{C} - \sigma_{C} \right\Vert_1 \ge \delta'_s, \]
for $\delta'_s$ strictly less than $\delta_s$, which we can enforce by setting 
$\delta'_s = \delta_s/\sqrt{2}$. Observe that $k$ is polynomial in the circuit size
because the number of parties cannot exceed the number of qubits upon which the circuit acts,
it is only linear in the number of qubits in system $C$, and the promise guarantees that
$1/\delta_s^2$ is polynomial in the circuit size. Then, using the same analysis as in the bipartite 
case in \eqref{eq:bipartite-NO-analysis-1}-\eqref{eq:bipartite-NO-analysis-3},
we have that
\begin{equation*}
\max_{\sigma_{C} \in \mathcal{E}_k} F(\rho_{C}, \sigma_{C}) \le 1 - \delta^2_s/8.
\end{equation*}

In the above we have obtained the same separation between completeness and soundness error
as in the bipartite case. As 
discussed in Section~\ref{sec:qsep-in-qip2}, this is sufficient to place the protocol in 
\textsf{QIP(2)}. (See Section 3.2 of Ref. \cite{JUW09} for how to amplify an inverse 
polynomial gap.) Thus, we have given a two-message quantum interactive proof 
system that decides the multipartite quantum separability problem.
\end{proof}

Note that \textsf{MULTI-QSEP-STATE}$_{\operatorname{1,1-LOCC}}$ is also \textsf{QSZK}- and \textsf{NP}-hard,
as the bipartite separability problem
is merely a special case of the multipartite separability problem.
Also, it should be clear that if we define a ``channel'' variant of 
\textsf{MULTI-QSEP-STATE}$_{\operatorname{1,1-LOCC}}$ that is the natural combination of \textsf{MULTI-QSEP-STATE}$_{\operatorname{1,1-LOCC}}$
and \textsf{QSEP-CHANNEL}$_{\operatorname{1,1-LOCC}}$, such a promise problem is \textsf{QIP}-complete by the analysis
in this and the previous section.

\section{Conclusion}

\label{sec:conclusion}We have provided the first nontrivial example of a
promise problem that is in \textsf{QIP(2)}\ and hard for both \textsf{QSZK}%
\ and \textsf{NP}. We accomplished this by introducing a version of the
quantum separability problem in which the input string is the specification of
a quantum circuit that generates a mixed bipartite state $\rho_{AB}$, along
with a promise that the state is close in trace distance to some separable
state or there is no separable state close to it in $1$-LOCC\ distance. We
showed that this promise problem, called \textsf{QSEP-STATE}$_{\operatorname{1,1-LOCC}}$, is decidable
by a two-message quantum interactive proof system, and we also proved that it
is hard for quantum statistical zero knowledge (\textsf{QSZK}) proof systems.
Our results in Section~\ref{sec:qsep-np-hard}\ also demonstrate that
\textsf{QSEP-STATE}$_{\operatorname{1,1-LOCC}}$ is hard for \textsf{NP} with respect
to Cook reductions. Finally, we considered a
natural variation of the quantum separability problem, in which the circuit
generates a channel rather than a state, and the goal is to determine if there
is an input to the circuit for which the output across some bipartite cut is
separable. The \textquotedblleft channel quantum separability
problem\textquotedblright\ is complete for \textsf{QIP}, the class of promise
problems decidable by general quantum interactive proof systems.
Furthermore, we have shown that a two-message quantum interactive proof system
can decide  a variant of the multipartite quantum separability problem in which 
the input and promises are similar to those from \textsf{QSEP-STATE}$_{\operatorname{1,1-LOCC}}$.

Some previous works have related separability testing to variants of the complexity
class \textsf{QMA}(2)
(quantum Merlin-Arthur proof systems with two unentangled provers)
 \cite{B10,HM13}.
Indeed, the class \textsf{QMA}(2) in some sense captures the
complexity of optimizing over the set of separable states (see Section~3.1 of
 \cite{HM13}---Section~4.1 in the arXiv version).
Although the variant of separability testing that we 
consider here is quite different, in both cases the problem
seems to be naturally captured by a quantum complexity class. However, in a follow-up paper,
we demonstrate that there {\it is} a variant of the quantum separability problems
considered here that is complete for the class \textsf{QMA}(2) \cite{MGHW13}, and thus
it is this latter problem that is related to the previous works mentioned above.

Not much is currently known about two-message quantum interactive proof
systems (\textsf{QIP(2)}), other than a nontrivial lower bound on it given by
Wehner \cite{W06}\ and the containment \textsf{QSZK~}$\subseteq~$%
\textsf{QIP(2)} \cite{W02,W09zkqa}. Also, \textsf{CLOSE-IMAGE} is a complete
promise problem for \textsf{QIP(2)}, but it really just amounts to a trivial
rewriting of the definition of \textsf{QIP(2)}, much like the relationship
between \textsf{CLOSE-IMAGES}\ and \textsf{QIP(3)} \cite{KW00,RW05,R09}. The
promise problem \textsf{QSEP-STATE}$_{\operatorname{1,1-LOCC}}$\ seems to have a natural two-message
quantum interative proof system that does not satisfy the zero-knowledge
property (if the verifier in our protocol were able to generate a
$k$-extension himself, then he would be able to solve the problem efficiently
in \textsf{BQP}, but our hardness results suggest that this should not be the
case). Thus, it seems possible that one might be able to show that
\textsf{QSEP-STATE}$_{\operatorname{1,1-LOCC}}$\ is in fact \textsf{QIP(2)}-complete, but a proof evades
us for now. It is nonetheless suggestive that a simple variation of \textsf{QSEP-STATE}$_{\operatorname{1,1-LOCC}}$
(\textsf{QSEP-CHANNEL}$_{\operatorname{1,1-LOCC}}$) leads to a \textsf{QIP}-complete promise problem.
Also, since it is unclear how to begin to show that \textsf{QSEP-STATE}$_{\operatorname{1,1-LOCC}}$\ is
\textsf{QIP(2)}-hard,\ it would be desirable to show that
\textsf{QSEP-STATE}$_{\operatorname{1,1-LOCC}}$\ is \textsf{QMA}-hard, given that \textsf{QMA~}%
$\subseteq~$\textsf{QIP(2)} and \textsf{QSZK~}$\subseteq~$\textsf{QIP(2)}, and
\textsf{QMA}\ and \textsf{QSZK}\ are not known to be commensurate complexity classes.

\section{Acknowledgements}

We acknowledge useful discussions with Fernando Brand{\~a}o, Anne Broadbent, Eric Chitambar, Claude
Cr\'epeau, Aram Harrow, and John Watrous. PH\ acknowledges support from the Canada Research
Chairs program, the Perimeter Institute, CIFAR, FQRNT's INTRIQ, NSERC and ONR
through grant N000140811249. The Perimeter Institute is supported by Industry
Canada and Ontario's Ministry of Economic Development \& Innovation.
MMW\ acknowledges support from the Centre de Recherches Math\'{e}matiques at
the University of Montreal.

\appendix

\section{Appendix}

\subsection{Approximate $k$-extendibility}

The following proposition applies to bipartite states that are approximately $k$-extendible:

\begin{proposition}
Let $\rho_{AB}$ be $\delta$-close$\ $to a $k$-extendible state, in the sense
that%
\begin{equation}
\min_{\sigma_{AB}\in\mathcal{E}_{k}}\left\Vert \rho_{AB}-\sigma_{AB}%
\right\Vert _{1}\leq\delta,\label{eq:delta-k-ext}%
\end{equation}
for some $\delta>0$, where $\mathcal{E}_{k}$ is the set of $k$-extendible
states. Then, the following bound holds%
\[
\left\Vert \rho_{AB}-\mathcal{S}\right\Vert _{\text{1-LOCC}}\leq\left(
16\ln2\left[  \frac{\log\left\vert A\right\vert }{k}\right]  \right)  ^{1/2} + \delta.
\]

\end{proposition}

\begin{proof}
The proof of this theorem requires only a simple modification of the original
proof of Brand\~{a}o \textit{et al}. \cite{BCY11a}. First, recall that the
squashed entanglement of some bipartite state $\omega_{AB}$ is equal to
$E_{\text{sq}}\left(  \omega_{A:B}\right)  \equiv\frac{1}{2}\inf\left\{  I\left(
A;B|E\right)  _{\omega^{\prime}}:\omega_{AB}=\text{Tr}_{E}\left\{
\omega_{ABE}\right\}  \right\}  $ \cite{CW04}. Let $\sigma_{AB}^{\prime}$ be
the state that achieves the minimum in (\ref{eq:delta-k-ext}), and let
$\sigma_{AB_{1}\cdots B_{k}}^{\prime}$ be a $k$-extension of $\sigma
_{AB}^{\prime}$. (The fact that such a minimum exists follows from the
argument above (\ref{eq:max-exists-arg}).)\ Following Brand\~{a}o \textit{et
al}., observe that%
\begin{align*}
\log\left\vert A\right\vert  &  \geq E_{\text{sq}}(\sigma_{A:B_{1}\cdots
B_{k}}^{\prime})\\
&  \geq\sum_{i=1}^{k}E_{\text{sq}}(\sigma_{A:B_{i}}^{\prime})\\
&  =kE_{\text{sq}}(\sigma_{A:B}^{\prime})\\
&  \geq k\left[  \frac{1}{16\ln2}\min_{\sigma_{AB}\in\mathcal{S}}\left\Vert
\sigma_{AB}^{\prime}-\sigma_{AB}\right\Vert _{\text{1-LOCC}}^{2}\right]  .
\end{align*}
The second inequality is from monogamy of squashed entanglement. The final
inequality exploits a theorem of Brand\~{a}o \textit{et al}. The above implies
that%
\[
\sqrt{16\ln2\left(  \frac{\log\left\vert A\right\vert }{k}\right)  } \geq\min
_{\sigma_{AB}\in\mathcal{S}}\left\Vert \sigma_{AB}^{\prime}-\sigma
_{AB}\right\Vert _{\text{1-LOCC}}.
\]
Let $\sigma_{AB}^{\ast}$ be the state achieving the minimum in $\min
_{\sigma_{AB}\in\mathcal{S}}\left\Vert \sigma_{AB}^{\prime}-\sigma
_{AB}\right\Vert _{\text{1-LOCC}}$. We then have that%
\begin{align*}
\left\Vert \sigma_{AB}^{\prime}-\sigma_{AB}^{\ast}\right\Vert _{\text{1-LOCC}%
}+\delta & \geq\left\Vert \sigma_{AB}^{\prime}-\sigma_{AB}^{\ast}\right\Vert
_{\text{1-LOCC}}+\left\Vert \sigma_{AB}^{\prime}-\rho_{AB}\right\Vert _{1}\\
& \geq\left\Vert \sigma_{AB}^{\prime}-\sigma_{AB}^{\ast}\right\Vert
_{\text{1-LOCC}}+\left\Vert \sigma_{AB}^{\prime}-\rho_{AB}\right\Vert
_{\text{1-LOCC}}\\
& \geq\left\Vert \rho_{AB}-\sigma_{AB}^{\ast}\right\Vert _{\text{1-LOCC}}\\
& \geq\min_{\sigma_{AB}\in\mathcal{S}}\left\Vert \rho_{AB}-\sigma
_{AB}\right\Vert _{\text{1-LOCC}}\ .
\end{align*}
From this, we conclude the statement of the proposition.
\end{proof}

The following proposition applies to $l$-partite states $\rho_C = \rho_{A_1:A_2:\cdots :A_l}$ that are approximately $k$-extendible:
\begin{proposition}
Let $\rho_{C}$ be $\delta$-close$\ $to a $k$-extendible state, in the sense
that%
\begin{equation} \label{eq:delta-multi-k-ext}
\min_{\sigma_{C}\in\mathcal{E}_{k}}\left\Vert \rho_{C}-\sigma_{C}%
\right\Vert _{1}\leq\delta,%
\end{equation}
for some $\delta>0$, where $\mathcal{E}_{k}$ is the set of $k$-extendible
$l$-partite states. Then, the following bound holds%
\[
\left\Vert \rho_{C}-\mathcal{S}\right\Vert _{\text{1-LOCC}}\leq\sqrt{\frac{4 l^2 \log\left\vert C\right\vert }{k - l}} + \delta
\]
where the quantity on the left is multipartite 1-LOCC distance (defined in (\ref{eq:multi-1-LOCC})) to the
set of fully separable states. 
\end{proposition}
\begin{proof}
Let $\sigma'_{C}$ be the state that achieves the minimum in (\ref{eq:delta-multi-k-ext}). Since this state is $k$-extendible, we have from Theorem~2 of \cite{BH12} that
\begin{equation} \label{eq:1locc-multi-sigma-min}
 \min_{\sigma_{C} \in \mathcal{S}} \left\Vert \sigma'_{C} - \sigma_{C} \right\Vert_{1-\textrm{LOCC}} \le \sqrt{\frac{4 l^2 \log\left\vert C\right\vert }{k - l}},
\end{equation} 
Let $\sigma^*_{C}$ be the state achieving the minimum on the left in (\ref{eq:1locc-multi-sigma-min}).
From the premise of the theorem, it follows that
\begin{align*}
\left\Vert \sigma'_{C} - \sigma^*_{C} \right\Vert_{1-\textrm{LOCC}} + \delta &> 
\left\Vert \sigma'_{C} - \sigma^*_{C} \right\Vert_{1-\textrm{LOCC}} + 
\left\Vert \sigma'_{C} - \rho_{C}\right\Vert_1\\
&\ge \left\Vert \sigma'_{C} - \sigma^*_{C} \right\Vert_{1-\textrm{LOCC}} + 
\left\Vert \sigma'_{C} - \rho_{C}\right\Vert_{1 - \textrm{LOCC}}\\
&\ge \left\Vert \sigma^*_{C} - \rho_{C} \right\Vert_{1-\textrm{LOCC}}\\
&\ge \min_{\sigma_{C} \in \mathcal{S}} \left\Vert \sigma_{C} - \rho_{C} \right\Vert_{1 - \textrm{LOCC}}.
\end{align*}
Thus,
\begin{align*}
\min_{\sigma_{C} \in \mathcal{S}} \left\Vert \sigma_{C} - \rho_{C} \right\Vert_{1 - \textrm{LOCC}} 
&< \left\Vert \sigma'_{C} - \sigma^*_{C} \right\Vert_{1-\textrm{LOCC}} + \delta\\
&\le \sqrt{\frac{4 l^2 \log\left\vert C\right\vert }{k - l}} + \delta,
\end{align*}
which concludes the proof.
\end{proof}

\bibliographystyle{alpha}
\bibliography{Ref}

\end{document}